\documentclass[a4paper,10pt,onecolumn]{IEEEtran}

\usepackage[utf8]{inputenc} 
\usepackage{hyperref}       
\usepackage{url}            
\usepackage{booktabs}       
\usepackage{amsfonts}       
\usepackage{nicefrac}       
\usepackage{xcolor}         
\usepackage{glossaries}
\usepackage{centernot}
\usepackage{aligned-overset}
\usepackage{amsthm}
\usepackage{bbm}
\usepackage{mathrsfs}
\usepackage{makecell}
\usepackage{bbm}
\usepackage{tikz}
\usepackage{pgfplots}
\usepackage{subcaption}

\newtheorem{theorem}{Theorem}
\newtheorem{lemma}{Lemma}

\newtheorem{definition}{Definition}
\newtheorem{remark}{Remark}
\newtheorem{problem}{Problem}

\newtheorem{example}{Example}

\newcommand{\noiselessSignal}{X}
\newcommand{\noisySignal}{\tilde{X}}
\newcommand{\estSignal}{\hat{X}}
\newcommand{\estSignalOpt}{X^*}
\newcommand{\noiselessSignalValue}{x}
\newcommand{\estSignalValue}{\hat{x}}
\newcommand{\noisySignalValue}{\tilde{x}}
\newcommand{\noiselessSignalAlphabet}{\mathcal{X}}
\newcommand{\noiselessSignalSigmaAlgebra}{\mathcal{F}}
\newcommand{\noiselessSignalAlphabetOneShot}{\mathbb{X}}
\newcommand{\noiselessSignalSigmaAlgebraOneShot}{\mathcal{G}}
\newcommand{\noisySignalAlphabet}{\tilde{\mathcal{X}}}
\newcommand{\noisySignalSigmaAlgebra}{\tilde{\mathcal{F}}}
\newcommand{\noisySignalAlphabetOneShot}{\tilde{\mathbb{X}}}
\newcommand{\noisySignalSigmaAlgebraOneShot}{\tilde{\mathcal{G}}}
\newcommand{\estSignalAlphabet}{\hat{\mathcal{X}}}
\newcommand{\estSignalSigmaAlgebra}{\hat{\mathcal{F}}}
\newcommand{\estSignalAlphabetOneShot}{\hat{\mathbb{X}}}
\newcommand{\estSignalSigmaAlgebraOneShot}{\hat{\mathcal{G}}}

\newcommand{\jointDist}{P}
\newcommand{\universalDist}{Q}
\newcommand{\seqLength}{T}
\newcommand{\seqIndex}{t}
\newcommand{\lossFunc}{\ell}
\newcommand{\maxLoss}{\ell_\textrm{max}}
\newcommand{\totalLoss}{L}
\newcommand{\maxTotalLoss}{L_\textrm{max}}
\newcommand{\Expectation}{\mathbb{E}}

\newcommand{\tvdist}[1]{\left\lVert #1 \right\rVert_\mathrm{TV}}
\newcommand{\kldiv}[2]{D\left({#1} || {#2}\right)}
\newcommand{\regret}{R}
\newcommand{\generalRVOne}{A}
\newcommand{\generalRVTwo}{B}
\newcommand{\generalRVOneValue}{a}
\newcommand{\generalRVTwoValue}{b}
\newcommand{\generalPmeasureOne}{\mu}
\newcommand{\generalPmeasureTwo}{\nu}
\newcommand{\rnderiv}[2]{\frac{d#1}{d#2}}
\newcommand{\rnderivInline}[2]{d#1/d#2}
\newcommand{\distributionIndex}{\theta}
\newcommand{\distributionIndexDomain}{\Theta}

\newcommand{\distributionIndexWeight}{w}
\newcommand{\diracDist}[1]{\delta_{#1}}

\newcommand{\indicator}[1]{\mathbbm{1}_{#1}}

\newcommand{\generalReal}{r}
\newcommand{\generalNatural}{n}
\newcommand{\naturals}{\mathbb{N}}

\newcommand{\cardinality}[1]{\left|#1\right|}
\newcommand{\absoluteValue}[1]{\left|#1\right|}
\newcommand{\typicalityParameter}{\varphi}
\newcommand{\typicalSet}{\mathcal{T}}
\newcommand{\countingFunction}{\#}
\newcommand{\complement}[1]{#1^c}
\newcommand{\jointDistMin}[1]{p_{\min, #1}}
\newcommand{\jointDistMinLen}[1]{p_{\min}(#1)}
\newcommand{\tvDistMaximizer}{A}
\newcommand{\problemTuple}{\mathcal{P}}
\newcommand{\reals}{\mathbb{R}}

\newcommand{\quantizer}{q}
\newcommand{\quantizationSetIntermediate}{S}
\newcommand{\quantizationSet}{A}
\newcommand{\quantizerLength}{\ell}
\newcommand{\floor}[1]{\left\lfloor #1 \right\rfloor}
\newcommand{\quantizerBit}{s}
\newcommand{\minset}{B}
\newcommand{\auxSequence}{a}
\newcommand{\distributionIndexIndex}{n}
\newcommand{\binomLikelihood}{\mathcal{L}}
\newcommand{\binomParameter}{p}
\newcommand{\binomTrials}{N}
\newcommand{\binomSuccesses}{n}

\newcommand{\eulersNumber}{e}
\newcommand{\exampleMixFunc}{f}
\newcommand{\regretPerComp}{\bar{R}}
\newcommand{\signalSubset}{A}
\newcommand{\bernoulli}[1]{\mathrm{Bernoulli}\left(#1\right)}
\newcommand{\cbdimension}{d}
\newcommand{\smallvalue}{\delta}
\newcommand{\dimensionIndexOne}{j}
\newcommand{\dimensionIndexTwo}{k}
\newcommand{\euclidnorm}[1]{\left\lVert #1 \right\rVert}
\newcommand{\fisherInformation}{I}
\newcommand{\landauo}{o}
\newcommand{\additiveNoise}{Z}
\newcommand{\additiveNoiseValue}{z}
\newcommand{\noiselessIntensity}{\lambda_X}
\newcommand{\noiseIntensity}{\lambda_Z}
\newcommand{\numNoiseless}{n_X}
\newcommand{\numNoise}{n_Z}
\newcommand{\trueNumNoiseless}{n_{X,0}}
\newcommand{\trueNumNoise}{n_{Z,0}}
\newcommand{\numNoiselessN}[1]{n_{X,#1}}
\newcommand{\numNoiseN}[1]{n_{Z,#1}}
\newcommand{\poissonpmf}[2]{\mathrm{Pois}(#1;#2)}
\newcommand{\entropy}[1]{H(#1)}
\newcommand{\generalIntensity}{\lambda}
\newcommand{\generalRVOneAlphabet}{\mathcal{A}}
\newcommand{\generalSubset}{S}

\DeclareMathOperator*{\argmax}{arg\,max}
\DeclareMathOperator*{\argmin}{arg\,min}

\newacronym{dude}{DUDE}{Discrete Universal Denoiser}
\newacronym{dmc}{DMC}{discrete memoryless channel}
\newacronym{dnn}{DNN}{deep neural network}
\newacronym{awgn}{AWGN}{additive white Gaussian noise}
\newacronym{AI}{AI}{artificial intelligence}

\definecolor{plot1}{HTML}{648fff}
\definecolor{plot2}{HTML}{785ef0}
\definecolor{plot3}{HTML}{dc267f}
\definecolor{plot4}{HTML}{fe6100}
\definecolor{plot5}{HTML}{ffb000}
\definecolor{plot6}{HTML}{000000}

\allowdisplaybreaks

\title{Universal Denoising without Channel Knowledge}
\author{Matthias Frey, Jonathan H. Manton, and Jingge Zhu\\Department of Electrical and Electronic Engineering, The University of Melbourne}

\begin{document}
\maketitle
\begin{abstract}
Inspired by a classical algorithm for online prediction, we propose a novel denoising scheme which is universal for families of probability distributions both in terms of the source that generates the noiseless signal and in terms of the channel that generates the noisy signal. The new denoising scheme does not rely on assumptions of exact channel knowledge or finite signal alphabets. Our analysis provides an upper bound for the performance gap compared with the Bayes envelope. For the special case in which the signal is an i.i.d. sequence and the family is countable, we characterize the consistency condition under which our scheme approaches the Bayes envelope as the length of the sequence tends to infinity. We also show that in general, approaching the Bayes envelope is not possible for a universal denoising scheme when this consistency condition is not satisfied. Furthermore, we show that, as in the online prediction case, the so-called plug-in approach (which relies on a maximum likelihood estimation of the underlying distribution parameter) does not approach the Bayes envelope in general. We also include numerical evaluations of our scheme for denoising of a Poisson signal.
\end{abstract}

\section{Introduction}
Denoising is a ubiquitous problem in signal processing. One of the most famous examples is the Kalman filter~\cite{kalman1960new} which has so many applications (cf., e.g., \cite{urrea2021kalman,auger2013industrial}) that it is a staple in signal processing textbooks. Although it is a computationally very efficient and in many applications highly effective denoising method, one of its central limitations is the need for the denoiser to know the statistics of the signal and noise. For the case in which knowledge about these statistics is limited, \emph{universal} denoising schemes such as the \gls*{dude}~\cite{weissman2005universal} have been proposed. Applications of universal denoisers include image denoising~\cite{ordentlich2003discrete,motta2010idude}, DNA sequence denoising~\cite{laehnemann2016denoising}, and denoising of medical measurements~\cite{huang2014efficient}. Most existing universal denoisers focus on the case of time-series signals (i.e., there is a correlation between the successive outputs of the source which is exploited for denoising), and typical limitations include that they are only applicable to the case of finite signal alphabets, that they require full knowledge of the noise channel, or that their performance does not approach the Bayes envelope.

In this work, we consider the following denoising problem: \emph{given a noisy signal $\noisySignal$, estimate the noiseless signal $\noiselessSignal$ without full knowledge of the underlying joint distribution $\jointDist_{\noiselessSignal, \noisySignal}$.} Specifically, we aim to find denoising strategies (i.e., functions that map $\noisySignal$ to estimates of $\noiselessSignal$) that are \emph{universal} for a family of distributions $(\jointDist_{\noiselessSignal, \noisySignal | \distributionIndex})_{\distributionIndex \in \distributionIndexDomain}$. That is, there is a denoising strategy which has, as long as $\jointDist_{\noiselessSignal, \noisySignal} = \jointDist_{\noiselessSignal, \noisySignal | \distributionIndex_0}$ for some $\distributionIndex_0 \in \distributionIndexDomain$, similar\footnote{See Definition~\ref{def:universal-denoising} for a rigorous definition.} performance as a hypothetical optimal denoiser\footnote{The performance of the optimal denoiser is also called the Bayes envelope.} that operates with full knowledge of $\jointDist_{\noiselessSignal, \noisySignal}$. To this end, we propose a novel universal denoising algorithm for i.i.d. signals which is inspired by a classical algorithm for online prediction problems~\cite{merhav1998universal}. This algorithm is universal in the joint distribution of the noiseless and noisy signals, which means that it relies neither on knowledge of the noise channel nor the statistics of the source of the noiseless signal. Furthermore, we prove performance guarantees that do not require assumptions of finite alphabet size.

\subsection{Prior Work}
\paragraph{Unknown source, known channel, finite alphabets}
To the best of our knowledge, the seminal work in the area of universal denoising is \cite{weissman2005universal} (with its earliest version published in \cite{weissman2002universal}). It introduces the sliding-window \gls*{dude} algorithm for denoising of finite-alphabet signals where the noisy signal is obtained by passing the noiseless signal through a \gls*{dmc} with transition probabilities known to the denoiser. Performance guarantees are given both for the semi-stochastic case where the noiseless signal is an individual sequence (i.e., it does not follow a probability distribution and could be adversarial with knowledge of the denoising strategy), and also for the stochastic case where the noiseless signal follows a probability distribution unknown to the denoiser. It is shown that in the stochastic case \gls*{dude} asymptotically achieves the Bayes envelope and in the semi-stochastic case it has at least as good worst-case performance as any alternative sliding-window denoiser of the same window size. There have been many extensions and follow-up works: for instance, \cite{ordentlich2009concentration} proposes a way of estimating the loss incurred by the denoiser, \cite{su2010universal} proposes a similar non-sliding window denoiser, and \cite{moon2016neural} proposes a variation that incorporates the use of a \gls*{dnn}. The authors of \cite{viswanathan2009lower} propose converse performance bounds for the \gls*{dude} setting. \gls*{dude} applications to denoising of biomedical data are explored in \cite{lee2017dude,fischer2019denoising}. Notably, the papers \cite{weissman2007universal,yan2026universal} make an explicit connection to the universal prediction problem.

\paragraph{Unknown channel and general alphabets}
While most of the literature on universal denoising is universal in the source only (i.e., there is an assumption that the noise comes from a \gls*{dmc} with transition probabilities known to the denoiser) and are restricted to the case of finite alphabets, there are a few works that relax these assumptions. The work \cite{dembo2005universal} assumes knowledge of the noise \gls*{dmc} but allows the alphabet of the noisy signal to be infinite while keeping the assumption that the alphabet of the noiseless signal is finite. The series of works \cite{gemelos2006universal,gemelos2006algorithms,gemelos2006universalPhD} proposes an algorithm for denoising of finite-alphabet signals under channel uncertainty. The assumption made for the scheme to work is that the possible channels are invertible which differs from the consistency condition we use in its place. The authors also show with an example that their assumptions are not sufficient to guarantee that the denoiser asymptotically achieves the Bayes envelope. In \cite{fozunbal2010regret}, both assumptions are relaxed, but the paper is focused on the case of \gls*{awgn} with unknown variance. The more recent paper \cite{notzel2017information} also does not assume knowledge of the noise statistics, but requires finite alphabets and focuses on a specific scenario of image denoising where multiple noisy versions of the same image are available.

\paragraph{i.i.d. signals}
Most existing literature is about denoising of time-series signals, i.e., there is usually some correlation between subsequent outputs of the source that produces the noiseless signal and these correlations are exploited for denoising. There are, however, a few works that focus on denoising of i.i.d. signals such as \cite{fozunbal2010regret} which assumes \gls*{awgn}, \cite{pereira2007denoising} which is not a universal algorithm as it assumes knowledge of both the source and noise statistics, and \cite{huang2014efficient} which assumes that the source distribution is approximately known. Notably, \cite{huang2014efficient} also gives a concrete application example for denoising of i.i.d. signals: medical measurements (e.g., blood pressure or white blood cell count) which are taken from different patients (i.e., there is no reason to believe that two measurements taken subsequently from different patients might be correlated).

\subsection{Contributions and Outline}
The contribution of our paper can be summarized as follows:
\begin{itemize}
  \item For the general form of the denoising problem, which we call the \emph{one-shot denoising problem}, we establish an information-theoretic performance guarantee for a denoising strategy that is mismatched to the true underlying distribution.
  \item For the special case of denoising of i.i.d. sequences, we propose a novel denoising strategy which is inspired by classical works on online prediction. Building on the results for one-shot denoising, we establish a sufficient condition for our strategy to approach the Bayes envelope.
  \item We show that in the case of countable parameter sets $\distributionIndexDomain$, this condition is also necessary.
  \item We show that there is a denoising problem in which the plug-in approach, a simpler alternative approach that uses a maximum likelihood estimate of the underlying distribution, does not approach the Bayes envelope while our proposed approach does.
  \item We conduct numerical evaluations of our proposed approach for the problem of denoising a Poisson signal.
\end{itemize}
The remainder of the paper is organized as follows: in Section~\ref{sec:one-shot}, we formally introduce the one-shot denoising problem and establish the associated performance guarantees. In Section~\ref{sec:iid-denoising}, we formally introduce the problem of denoising of i.i.d. sequences and the notion of universality. Performance guarantees for sufficient conditions are established in Section~\ref{sec:iid-countable} for countable parameter sets and Section~\ref{sec:iid-uncountable} for uncountable parameter sets, respectively. The necessity of the sufficient condition in the countable case is shown in Section~\ref{sec:iid-necessary}. In Section~\ref{sec:plug-in}, we study the plug-in approach, and in Section~\ref{sec:poisson}, we numerically evaluate our approach for the problem of denoising a Poisson signal.

\subsection{Notation}
We use capital letters such as $\generalRVOne$ and $\generalRVTwo$ to denote random variables and small letters such as $\generalRVOneValue$ and $\generalRVTwoValue$ to denote values that they can take. Given probability measures $\generalPmeasureOne$ and $\generalPmeasureTwo$, we use $\Expectation_\generalPmeasureOne$ to denote the expectation with respect to $\generalPmeasureOne$ and when $\generalPmeasureOne$ is absolutely continuous with respect to $\generalPmeasureTwo$, we use $\rnderivInline{\generalPmeasureOne}{\generalPmeasureTwo}$ to denote the Radon-Nikodym derivative. The Kullback-Leibler divergence is denoted $\kldiv{\generalPmeasureOne}{\generalPmeasureTwo} := \Expectation_\generalPmeasureOne \rnderivInline{\generalPmeasureOne}{\generalPmeasureTwo}$ with the usual convention that Kullback-Leibler divergence is infinite whenever the Radon-Nikodym derivative does not exist. If $\generalRVOne$ is a random variable with probability distribution described by $\generalPmeasureOne_\generalRVOne$, we use $\entropy{\generalPmeasureOne_\generalRVOne}$ to denote the Shannon entropy of $\generalRVOne$ when it is distributed according to $\generalPmeasureOne_\generalRVOne$ (note that we do not use the more common notation $\entropy{\generalRVOne}$ because it is ambiguous when the distribution of $\generalRVOne$ is not clear from context). If $\generalPmeasureOne$ is any signed measure on the measurable space $\generalRVOneAlphabet$, we use $\tvdist{\generalPmeasureOne} := \max\{\generalPmeasureOne(\generalSubset):~ \generalSubset \subseteq \generalRVOneAlphabet \text{ is measurable}\}$. If $\generalRVOneValue \in \generalRVOneAlphabet$ where $\generalRVOneAlphabet$ is a measurable space, we use $\diracDist{\generalRVOneValue}$ to denote the Dirac measure with probability mass at $\generalRVOneValue$, i.e., sets have measure $1$ if they contain $\generalRVOneValue$ and they have measure $0$ otherwise. If $\generalSubset$ is a subset of some set $\generalRVOneAlphabet$ which is clear from context, we use $\complement{\generalSubset} := \generalRVOneAlphabet \setminus \generalSubset$ to denote the complement of $\generalSubset$. We also use
\[
\indicator{\generalSubset}:~
\generalRVOneAlphabet \rightarrow \{0,1\},~
\generalRVOneValue \mapsto
\begin{cases}
  1, &\generalRVOneValue \in \generalSubset,\\
  0, &\text{otherwise,}
\end{cases}
\]
to denote the indicator function. Elements of the product space $\generalRVOneAlphabet^\generalNatural$ are denoted $\generalRVOneValue^\generalNatural$ and their $\generalNatural'$-th component is denoted $\generalRVOneValue_{\generalNatural'}$. For a set $\generalRVOneAlphabet$ and $\generalNatural \in \naturals$ (which will be clear from context), we define $\countingFunction: \generalRVOneAlphabet^\generalNatural \times \generalRVOneAlphabet \rightarrow \naturals, (\generalRVOneValue^\generalNatural, \generalRVOneValue) \mapsto \cardinality{\{\seqIndex: \generalRVOneValue_\seqIndex = \generalRVOneValue\}}$ (where $\cardinality{\generalSubset}$ is the number of elements in a finite set $\generalSubset$) to denote a function that counts the numbers of entries of a tuple.

\section{One-Shot Denoising}
\label{sec:one-shot}
In this section, we introduce the problem of removing noise from a single observation $\noisySignal$ to obtain an estimate of the noiseless signal $\noiselessSignal$. Since $\noisySignal$ and $\noiselessSignal$ can be vector-valued random variables, this is a more general problem than denoising of sequences of observations. We also present performance results for this general scenario. While the generality of the problem statement inherently means that the strength of the results is somewhat limited, they nonetheless form the basis for the main results of our paper on denoising of i.i.d. sequences which we present in Section~\ref{sec:iid-denoising}.

\begin{problem}
\label{problem:oneshot}
\emph{(One-shot denoising).}
A \emph{one-shot denoising problem} $\problemTuple = \big( (\noiselessSignalAlphabetOneShot, \noiselessSignalSigmaAlgebraOneShot), (\noisySignalAlphabetOneShot, \noisySignalSigmaAlgebraOneShot), (\estSignalAlphabetOneShot, \estSignalSigmaAlgebraOneShot), \jointDist_{\noiselessSignal, \noisySignal}, \totalLoss \big)$ is characterized by measurable spaces $(\noiselessSignalAlphabetOneShot, \noiselessSignalSigmaAlgebraOneShot)$, the \emph{noiseless signal alphabet}, $(\noisySignalAlphabetOneShot, \noisySignalSigmaAlgebraOneShot)$, the \emph{noisy signal alphabet}, and $(\estSignalAlphabetOneShot, \estSignalSigmaAlgebraOneShot)$ , the \emph{estimated signal alphabet}, as well as a probability distribution $\jointDist_{\noiselessSignal, \noisySignal}$ on $(\noiselessSignalAlphabetOneShot \times \noisySignalAlphabetOneShot, \noiselessSignalSigmaAlgebraOneShot \otimes \noisySignalSigmaAlgebraOneShot)$ and a measurable \emph{loss function} $\totalLoss: \estSignalAlphabetOneShot \times \noiselessSignalAlphabetOneShot \rightarrow \reals$.

It consists of the following steps:
\begin{itemize}
  \item Nature draws the random variables $\noiselessSignal, \noisySignal$ from $\jointDist_{\noiselessSignal, \noisySignal}$.
  \item The denoiser observes only $\noisySignal$ and computes an estimate $\estSignal \in \estSignalAlphabetOneShot$ of the unknown noiseless signal $\noiselessSignal$.
  \item The denoiser incurs the loss $\totalLoss(\estSignal, \noiselessSignal)$.
\end{itemize}
\end{problem}
In this paper, we consider two cases for what the loss function can be, each yielding slightly different results.
\begin{definition}\emph{(Loss functions).}
\label{def:loss}
\begin{enumerate}
\item We say that $\totalLoss$ is \emph{bounded} by $\maxTotalLoss$ if it only takes values in $[0,\maxTotalLoss]$.
\item We say that $\totalLoss$ is the \emph{self-information loss} if $\noiselessSignalAlphabetOneShot$ is countable, $\estSignalAlphabetOneShot$ is the set of p.m.f.s on $\noiselessSignalAlphabetOneShot$, and
\[
\totalLoss(\estSignalValue, \noiselessSignalValue) = - \log \estSignalValue(\noiselessSignalValue).
\]
\end{enumerate}
\end{definition}

Throughout this paper, we use $\jointDist_{\noiselessSignal}$, respectively $\jointDist_{\noisySignal}$ to denote the corresponding marginals of $\jointDist_{\noiselessSignal, \noisySignal}$. Furthermore, we assume that regular conditional probabilities $\jointDist_{\noisySignal | \noiselessSignal}$ and $\jointDist_{\noiselessSignal | \noisySignal}$ exist. A \emph{denoising strategy} is a function that maps the noisy signal $\noisySignal$ to an estimate $\estSignal$. We compare the denoiser's performance with that of a strategy which can utilize full knowledge of $\jointDist_{\noiselessSignal | \noisySignal}$ and optimizes for the lowest possible expected loss. That is, we define an optimal estimate (see Remark~\ref{remark:measurability} for sufficient conditions for existence)
\begin{equation}
\label{eq:optimal-denoising}
\estSignalOpt
\in
\argmin_{\estSignalValue \in \estSignalAlphabetOneShot}
  \Expectation_{\jointDist_{\noiselessSignal | \noisySignal}} \totalLoss(\estSignalValue, \noiselessSignal).
\end{equation}
The quantity
\begin{equation*}
\Expectation_{\jointDist_{\noiselessSignal | \noisySignal}}
  \totalLoss(\estSignalOpt, \noiselessSignal)
=
\min_{\estSignalValue \in \estSignalAlphabetOneShot}
  \Expectation_{\jointDist_{\noiselessSignal | \noisySignal}} \totalLoss(\estSignalValue, \noiselessSignal),
\end{equation*}
which is a function of $\noisySignal$, is called the \emph{Bayes envelope}. We will evaluate the performance of the denoisers considered in this paper in terms of how closely their performance approaches the Bayes envelope. Specifically, we consider a denoising strategy that is induced by a probability distribution $\universalDist_{\noiselessSignal, \noisySignal}$. Under this strategy, the denoiser chooses an estimate that minimizes the loss in the case that $\jointDist_{\noiselessSignal, \noisySignal} = \universalDist_{\noiselessSignal, \noisySignal}$, i.e., the estimate $\estSignal$ satisfies
\begin{equation}
\label{eq:mismatched-denoising}
\estSignal
\in
\argmin_{\estSignalValue \in \estSignalAlphabetOneShot}
  \Expectation_{\universalDist_{\noiselessSignal | \noisySignal}} \totalLoss(\estSignalValue, \noiselessSignal).
\end{equation}

\begin{remark}
\label{remark:measurability}
\emph{(Existence and measurability of estimates).}
In order for the expectations involving $\estSignalOpt$ and $\estSignal$ (which will appear throughout the remainder of the paper) to be well-defined, it is important to ensure that they exist and are random variables. Since $\noisySignal$ is a random variable, it is sufficient to make sure that they are measurable functions of $\noisySignal$. In general, we believe that this is only possible by imposing some additional conditions on the measurable space $(\estSignalAlphabetOneShot, \estSignalSigmaAlgebraOneShot)$. The measurable maximum theorem~\cite[Theorem 18.19]{aliprantis2006infinite} tells us that it is for instance sufficient to assume that $(\estSignalAlphabetOneShot, \estSignalSigmaAlgebraOneShot)$ is metrizable such that all of the following hold:
\begin{itemize}
  \item $(\estSignalAlphabetOneShot, \estSignalSigmaAlgebraOneShot)$ is separable and compact.
  \item For every $\estSignalValue \in \estSignalAlphabetOneShot$, the function $\totalLoss(\estSignalValue, \cdot)$ is measurable.
  \item For every $\noiselessSignalValue \in \noiselessSignalAlphabetOneShot$, the function $\totalLoss(\cdot,\noiselessSignalValue)$ is continuous.
\end{itemize}
An important special case is when $\estSignalAlphabetOneShot$ is finite and $\estSignalSigmaAlgebraOneShot$ is the discrete $\sigma$-algebra. For the remainder of this paper, we assume tacitly that measurable $\estSignalOpt$ and $\estSignal$ exist which almost surely satisfy \eqref{eq:optimal-denoising} and \eqref{eq:mismatched-denoising}.
\end{remark}

Since in general, we have $\jointDist_{\noiselessSignal, \noisySignal} \neq \universalDist_{\noiselessSignal, \noisySignal}$, we can also expect that usually $\Expectation_{\jointDist_{\noiselessSignal | \noisySignal}} \totalLoss(\estSignal, \noiselessSignal) > \Expectation_{\jointDist_{\noiselessSignal | \noisySignal}} \totalLoss(\estSignalOpt, \noiselessSignal)$, i.e., the loss incurred by the predictor is greater than the Bayes envelope. Following the terminology of the online learning literature, we define the mean difference of these two quantities as the \emph{expected regret} $\regret$ of the denoiser:
\begin{equation}
\label{eq:def-regret}
\regret
:=
\Expectation_{\jointDist_{\noiselessSignal, \noisySignal}}\big(
  \totalLoss(\estSignal, \noiselessSignal)
  -
  \totalLoss(\estSignalOpt, \noiselessSignal)
\big).
\end{equation}
For the one-shot denoising problem and denoising strategies induced by a probability distribution $\universalDist_{\noiselessSignal, \noisySignal}$, we have the following result on expected regret. The proof of Theorem~\ref{theorem:oneshot} is relegated to Appendix~\ref{appendix:oneshot-proof}.
\begin{theorem}
\emph{(Sufficient conditions for one-shot denoising).}
\label{theorem:oneshot}
Consider the setting of Problem~\ref{problem:oneshot}. If $\totalLoss$ is bounded by $\maxTotalLoss$ for some $\totalLoss_{\max} \in [0,\infty)$, we have
\begin{equation}
\label{eq:oneshot-bounded-loss}
\regret
\leq
\totalLoss_{\max}
\sqrt{
  \frac{1}{2}\left(
    \kldiv{\jointDist_{\noiselessSignal, \noisySignal}}{\universalDist_{\noiselessSignal, \noisySignal}}
    -
    \kldiv{\jointDist_{\noisySignal}}{\universalDist_{\noisySignal}}
  \right).
}
\end{equation}
If $\totalLoss$ is the self-information loss, we have
\begin{equation}
\label{eq:oneshot-log-loss}
\regret
=
\kldiv{\jointDist_{\noiselessSignal, \noisySignal}}{\universalDist_{\noiselessSignal, \noisySignal}}
-
\kldiv{\jointDist_{\noisySignal}}{\universalDist_{\noisySignal}}.
\end{equation}
\end{theorem}
As we might expect intuitively, Theorem~\ref{theorem:oneshot} states that if $\universalDist_{\noiselessSignal,\noisySignal}$ is not much more dissimilar from $\jointDist_{\noiselessSignal,\noisySignal}$ than $\universalDist_{\noisySignal}$ is from $\jointDist_{\noisySignal}$, the regret will be small. The theorem quantifies this exactly in terms of well-known information-theoretic quantities. Another way to look at this is via the identity
\[
\kldiv{\jointDist_{\noiselessSignal, \noisySignal}}{\universalDist_{\noiselessSignal, \noisySignal}}
-
\kldiv{\jointDist_{\noisySignal}}{\universalDist_{\noisySignal}}
=
\Expectation_{\jointDist_{\noisySignal}}
  \kldiv{\jointDist_{\noiselessSignal | \noisySignal}}{\universalDist_{\noiselessSignal | \noisySignal}}
\]
which holds by Lemma~\ref{lemma:conditional-kl} in Appendix~\ref{appendix:oneshot-proof}. From this point of view, Theorem~\ref{theorem:oneshot} quantifies the average regret in terms of how different the conditional distributions of the noiseless signal are on average. As can be seen in the proof of Theorem~\ref{theorem:oneshot} in Appendix~\ref{appendix:oneshot-proof}, it is a relatively straightforward result, and it also leaves the question open how the denoiser should choose a $\universalDist_{\noiselessSignal,\noisySignal}$ that leads to a small expected regret. Answering these questions may not be possible without considering special characteristics of the denoising problem at hand. In this paper we will answer these questions for a class of denoising problems involving i.i.d. sequences.

\section{Denoising of i.i.d. Sequences}
\label{sec:iid-denoising}
In this section, we specialize the one-shot denoising problem to the case of i.i.d. sequences. The idea behind this is that the i.i.d. nature of the sequence allows us to draw conclusions about the distribution of one component from observing the other components and thereby enables us to realize smaller regrets as the sequence length grows. In this case, we can construct a denoising strategy that works well as long as the true distribution of signals is contained in a family of distributions known a priori to the denoiser. One important (but intuitively obvious) caveat is that this family cannot contain two distributions that are very different in terms of noiseless signals but indistinguishable in terms of noisy signals. This consistency condition is formally stated in \eqref{eq:consistency} and its necessity is illustrated with an example in Section~\ref{sec:iid-necessary}. We consider a bounded additive loss per component as well as the self-information loss.
\begin{problem}
\emph{(i.i.d. denoising problem).}
\label{problem:iid}
An \emph{i.i.d. denoising problem} $\problemTuple = \big( (\noiselessSignalAlphabet, \noiselessSignalSigmaAlgebra), (\noisySignalAlphabet, \noisySignalSigmaAlgebra), (\estSignalAlphabet, \estSignalSigmaAlgebra), \jointDist_{\noiselessSignal, \noisySignal}, \totalLoss \big)$ is characterized by measurable spaces $(\noiselessSignalAlphabet, \noiselessSignalSigmaAlgebra)$, the \emph{noiseless signal alphabet}, $(\noisySignalAlphabet, \noisySignalSigmaAlgebra)$, the \emph{noisy signal alphabet}, and $(\estSignalAlphabet, \estSignalSigmaAlgebra)$, the \emph{estimated signal alphabet}, as well as a probability distribution $\jointDist_{\noiselessSignal, \noisySignal}$ on $(\noiselessSignalAlphabet \times \noisySignalAlphabet, \noiselessSignalSigmaAlgebra \otimes \noisySignalSigmaAlgebra)$ and a measurable \emph{loss function} $\totalLoss: \estSignalAlphabet^\seqLength \times \noiselessSignalAlphabet^\seqLength \rightarrow \reals$.

Its instantiation for \emph{sequence length} $\seqLength$ consists of the following steps:
\begin{itemize}
  \item Nature draws a sequence $(\noiselessSignal_1, \noisySignal_1), \dots, (\noiselessSignal_\seqLength, \noisySignal_\seqLength)$ i.i.d. from $\jointDist_{\noiselessSignal, \noisySignal}$.
  \item The denoiser observes only $\noisySignal^\seqLength$ and computes an estimate $\estSignal^\seqLength \in \estSignalAlphabet^\seqLength$ of the unknown noiseless signal $\noiselessSignal^\seqLength$.
  \item The denoiser incurs the loss $\totalLoss(\estSignal^\seqLength, \noiselessSignal^\seqLength)$.
\end{itemize}
\end{problem}
Given an i.i.d. denoising problem $\problemTuple = \big( (\noiselessSignalAlphabet, \noiselessSignalSigmaAlgebra), (\noisySignalAlphabet, \noisySignalSigmaAlgebra), (\estSignalAlphabet, \estSignalSigmaAlgebra), \jointDist_{\noiselessSignal, \noisySignal}, \totalLoss \big)$, each sequence length $\seqLength \in \naturals$ induces a one-shot denoising problem $\problemTuple_\seqLength := \big( (\noiselessSignalAlphabet^\seqLength, \noiselessSignalSigmaAlgebra^{\otimes \seqLength}), (\noisySignalAlphabet^\seqLength, \noisySignalSigmaAlgebra^{\otimes \seqLength}), (\estSignalAlphabet^\seqLength, \estSignalSigmaAlgebra^{\otimes \seqLength}), \jointDist_{\noiselessSignal, \noisySignal}^\seqLength, \totalLoss \big)$, where $\jointDist_{\noiselessSignal, \noisySignal}^\seqLength$ is the $\seqLength$-fold product of $\jointDist_{\noiselessSignal, \noisySignal}$. Clearly, $\problemTuple$ at sequence length $\seqLength$ and $\problemTuple_\seqLength$ describe the same random experiment. We define the \emph{average expected per-component regret} as
\[
  \regretPerComp := \frac{\regret}{\seqLength},
\]
where $\regret$ is the expected regret incurred by the denoiser in the induced one-shot denoising problem. We say that the loss is \emph{additive} if there is a per-component loss function $\lossFunc: \estSignalAlphabet \times \noiselessSignalAlphabet \rightarrow [0,\infty)$ such that $\totalLoss(\estSignalValue^\seqLength, \noiselessSignalValue^\seqLength) = \sum_{\seqIndex=1}^\seqLength \lossFunc(\estSignalValue_\seqIndex, \noiselessSignalValue_\seqIndex)$ for all $\estSignalValue^\seqLength \in \estSignalAlphabet^\seqLength, \noiselessSignalValue^\seqLength \in \noiselessSignalAlphabet^\seqLength$. Analogously to Definition~\ref{def:loss}, we say that $\lossFunc$ is bounded by $\maxLoss$ if it only takes values in $[0,\maxLoss]$. Clearly, in the case of an additive loss, $\lossFunc$ is bounded by $\maxLoss$ if and only if $\totalLoss$ is bounded by $\seqLength\maxLoss$.

We next construct a denoising strategy which is universal for a (possibly very large) family $
\left(\jointDist_{\noiselessSignal,\noisySignal|\distributionIndex}\right)_{\distributionIndex \in \distributionIndexDomain}$ of probability distributions on $\noiselessSignalAlphabet \times \noisySignalAlphabet$ where $\distributionIndexDomain$ is a measurable space.

\begin{definition}
\label{def:universal-denoising}
\emph{(Universal denoising).}
A denoising strategy is called \emph{universal} if it works well as long as $\jointDist_{\noiselessSignal,\noisySignal} = \jointDist_{\noiselessSignal,\noisySignal|\distributionIndex_0}$ for some $\distributionIndex_0 \in \distributionIndexDomain$. The definition of ``works well'' in this context can differ, but in this work we will consider two possibilities:
\begin{enumerate}
  \item\label{item:universal-denoising-total} $\regret \rightarrow 0$ as $\seqLength \rightarrow \infty$,
  \item\label{item:universal-denoising-perComp} $\regretPerComp \rightarrow 0$ as $\seqLength \rightarrow \infty$.
\end{enumerate}
\end{definition}
Obviously \ref{item:universal-denoising-total} yields a stronger version of Definition~\ref{def:universal-denoising} in the sense that it implies \ref{item:universal-denoising-perComp} (but not necessarily vice versa).

We construct the universal denoising strategy in a way similar as is done for universal prediction~\cite{merhav1998universal}. Namely, we define some probability measure $\distributionIndexWeight$ on $\distributionIndexDomain$ (as can be seen in the results, it does not usually matter too much what $\distributionIndexWeight$ is chosen as long as it is supported on all of $\distributionIndexDomain$) and use
\begin{equation}
\label{eq:universal-dist}
  \universalDist_{\noiselessSignal^\seqLength, \noisySignal^\seqLength}(\signalSubset)
  :=
  \int_\distributionIndexDomain
    \jointDist_{\noiselessSignal, \noisySignal | \distributionIndex}^\seqLength(\signalSubset)
    \distributionIndexWeight(d\distributionIndex)
\end{equation}
with the denoising strategy \eqref{eq:mismatched-denoising}.

\subsection{Sufficient conditions and performance guarantees for countable $\distributionIndexDomain$}
\label{sec:iid-countable}
In this subsection, we establish our main result on denoising of i.i.d. sequences for the case in which the parametric family of possible distributions is countable.

\begin{theorem}
\label{theorem:iid-denoising}
\emph{(Sufficient conditions for i.i.d. denoising with countable parameter space).}
Let $(\jointDist_{\noiselessSignal, \noisySignal | \distributionIndex})_{\distributionIndex \in \distributionIndexDomain}$ be a parametric family of probability distributions on $\noiselessSignalAlphabet \times \noisySignalAlphabet$ where $\distributionIndexDomain$ is countable, and assume that it satisfies
\begin{equation}
\label{eq:consistency}
  \forall \distributionIndex_1 \neq \distributionIndex_2 \in \distributionIndexDomain
  :~
  \jointDist_{\noisySignal | \distributionIndex_1}
  =
  \jointDist_{\noisySignal | \distributionIndex_2}
  \Rightarrow
  \jointDist_{\noiselessSignal,\noisySignal | \distributionIndex_1}
  =
  \jointDist_{\noiselessSignal,\noisySignal | \distributionIndex_2}
  .
\end{equation}
Assume further that $\jointDist_{\noiselessSignal^\seqLength, \noisySignal^\seqLength} = \jointDist_{\noiselessSignal, \noisySignal}^\seqLength = \jointDist_{\noiselessSignal, \noisySignal | \distributionIndex_0}^\seqLength$ for some $\distributionIndex_0$ in the support of a p.m.f. $\distributionIndexWeight$ on $\distributionIndexDomain$, and let $\universalDist_{\noiselessSignal^\seqLength, \noisySignal^\seqLength}(\signalSubset)$ be defined as in \eqref{eq:universal-dist}.
\begin{enumerate}
  \item\label{item:iid-denoising-finite} If $\noisySignalAlphabet$ and $\distributionIndexDomain$ are finite, then for large enough $\seqLength$, we have
  \begin{align*}
    \regretPerComp &\leq \maxLoss \exp\left(-\frac{1}{2}\sqrt{\seqLength}\right) &\text{if the loss is additive with $\lossFunc$ bounded by $\maxLoss$,}\\
    \regret &\leq \exp\left(-\sqrt{\seqLength}\right) &\text{if $\lossFunc$ is the self-information loss.}
  \end{align*}
  \item\label{item:iid-denoising-countable} If $\noisySignalAlphabet$ is a general measurable space and $\distributionIndexDomain$ is finite or countably infinite, we have
  \begin{align*}
    \lim_{\seqLength \rightarrow \infty} \regretPerComp &= 0 &\text{if the loss is additive with $\lossFunc$ bounded by $\maxLoss$,}\\
    \lim_{\seqLength \rightarrow \infty} \regret &= 0 &\text{if $\lossFunc$ is the self-information loss.}
  \end{align*}
\end{enumerate}
\end{theorem}

This theorem follows directly from Theorem~\ref{theorem:oneshot} and the following lemma.

\begin{lemma}
\label{lemma:kldiv-iid}
Let $(\jointDist_{\noiselessSignal, \noisySignal | \distributionIndex})_{\distributionIndex \in \distributionIndexDomain}$ be a parametric family of probability distributions on $\noiselessSignalAlphabet \times \noisySignalAlphabet$ and assume that $\distributionIndexDomain$ is countable, $\jointDist_{\noiselessSignal^\seqLength, \noisySignal^\seqLength} = \jointDist_{\noiselessSignal, \noisySignal}^\seqLength = \jointDist_{\distributionIndex_0, \noiselessSignal, \noisySignal}^\seqLength$ for some $\distributionIndex_0 \in \distributionIndexDomain$ with $\distributionIndexWeight(\distributionIndex_0) > 0$, condition \eqref{eq:consistency} is satisfied, and $\universalDist_{\noiselessSignal^\seqLength, \noisySignal^\seqLength}$ is given by \eqref{eq:universal-dist}.

\begin{enumerate}
  \item\label{item:iid-lemma-finite} If $\noisySignalAlphabet$ and $\distributionIndexDomain$ are finite, then for large enough $\seqLength$, we have
  \[
    \kldiv{\jointDist_{\noiselessSignal, \noisySignal}^\seqLength}{\universalDist_{\noiselessSignal^\seqLength, \noisySignal^\seqLength}}
    -
    \kldiv{\jointDist_{\noisySignal}^\seqLength}{\universalDist_{\noisySignal^\seqLength}}
    \leq
    \exp\left(-\sqrt{\seqLength}\right).
  \]
  \item\label{item:iid-lemma-countable} If $\noisySignalAlphabet$ is a general measurable space and $\distributionIndexDomain$ is countable, we have
  \[
  \lim_{\seqLength \rightarrow \infty}\left(
    \kldiv{\jointDist_{\noiselessSignal, \noisySignal}^\seqLength}{\universalDist_{\noiselessSignal^\seqLength, \noisySignal^\seqLength}}
    -
    \kldiv{\jointDist_{\noisySignal}^\seqLength}{\universalDist_{\noisySignal^\seqLength}}
  \right)
  =
  0.
  \]
\end{enumerate}
\end{lemma}
The proof of Lemma~\ref{lemma:kldiv-iid} is relegated to Appendix~\ref{appendix:iid-proof}, but we provide a sketch here.

In order to establish Item~\ref{item:iid-lemma-finite}, we use the upper bound
\[
\kldiv{\jointDist_{\noiselessSignal, \noisySignal}^\seqLength}{\universalDist_{\noiselessSignal^\seqLength, \noisySignal^\seqLength}}
\leq
\log\frac{1}{\distributionIndexWeight(\distributionIndex_0)}
\]
which trivially follows from \eqref{eq:universal-dist} and split the expectation that appears in the definition of $-\kldiv{\jointDist_{\noisySignal}^\seqLength}{\universalDist_{\noisySignal^\seqLength}}$ into a typical and an atypical part. We use letter typicality, which is based on the relative frequency of letters from $\noisySignalAlphabet$ appearing in $\noisySignal^\seqLength$ and is also the reason that we need $\noisySignalAlphabet$ to be finite. The atypical part of the expectation can be bounded with standard techniques (essentially, this part of the expectation is small because the probability of drawing an atypical sequence is small). In order to bound the typical part of the expectation, we upper bound the probability of the typical set under $\jointDist_{\noisySignal | \distributionIndex}$ for $\distributionIndex \neq \distributionIndex_0$. While this approach is standard in principle as well, a slight departure from usual methods is required to establish this bound for the probability of the typical set. In the process, we make use of the fact that the finiteness of $\distributionIndexDomain$ implies that the variational distance between $\jointDist_{\noisySignal | \distributionIndex}$ and $\jointDist_{\noisySignal | \distributionIndex_0}$ is uniformly bounded away from $0$ for all $\distributionIndex \neq \distributionIndex_0$.

In order to establish Item~\ref{item:iid-denoising-countable}, we apply a successively finer quantization of $\noisySignalAlphabet$ which depends on $\distributionIndex_0$ and the parametric family $(\jointDist_{\noiselessSignal, \noisySignal | \distributionIndex})_{\distributionIndex \in \distributionIndexDomain}$. It is important to note that this does \emph{not} mean that the denoising strategy requires knowledge of $\distributionIndex_0$ since the quantization is only used in the proof that the regret is small but not in the description of the denoising strategy. The advantage of this parameter-dependent quantization strategy is that we do not require any regularity conditions on $\noisySignalAlphabet$ other than that it is a measurable space.

\begin{remark}
The requirement in Theorem~\ref{theorem:iid-denoising} that the loss is additive with $\lossFunc$ bounded by $\maxLoss$ can be relaxed to only require that $\totalLoss$ is bounded by $\seqLength \maxLoss$. However, additive losses are in our opinion the most relevant case in which we have a bound of this form. We will also see in Lemma~\ref{lemma:simplified-additive} that it is significantly easier to compute \eqref{eq:mismatched-denoising} if the loss is additive.
\end{remark}

\subsection{Sufficient conditions and performance guarantees for uncountable $\distributionIndexDomain$}
\label{sec:iid-uncountable}
For the case of uncountable parametric families of distributions, we focus on the case in which both the parameter set and the signal alphabets are continuous or finite and the parametric family satisfies smoothness and regularity conditions which are common in parametric statistics.

\begin{theorem}
\label{theorem:iid-denoising-continuous}
\emph{(Sufficient conditions for i.i.d. denoising with continuous parameter space).}
Let $\noisySignalAlphabet$ and $\noiselessSignalAlphabet$ be Euclidean or finite spaces, let $(\jointDist_{\noiselessSignal, \noisySignal | \distributionIndex})_{\distributionIndex \in \distributionIndexDomain}$ be a parametric family of probability distributions on $\noiselessSignalAlphabet \times \noisySignalAlphabet$ described by densities or p.m.f.s\footnote{We will use symbols such as $\jointDist_{\noiselessSignal, \noisySignal | \distributionIndex}$ interchangeably for the probability measure and the density or p.m.f. it is described by.} where $\distributionIndexDomain \subseteq \reals^\cbdimension$, and assume that it satisfies the following regularity conditions:
\begin{enumerate}
  \item\label{item:cb-smoothness-density} $\jointDist_{\noiselessSignal, \noisySignal | \distributionIndex}(\noiselessSignalValue, \noisySignalValue)$ and $\jointDist_{\noisySignal | \distributionIndex}(\noisySignalValue)$ are twice continuously differentiable at $\distributionIndex_0$ as functions of $\distributionIndex$ for almost all $\noiselessSignalValue, \noisySignalValue$, and there exists $\smallvalue > 0$ such that for all $\dimensionIndexOne, \dimensionIndexTwo \in \{1, \dots, \cbdimension\}$
  \begin{align*}
    \Expectation_{\jointDist_{\noiselessSignal, \noisySignal | \distributionIndex_0}}
      \sup_{\distributionIndex: \euclidnorm{\distributionIndex-\distributionIndex_0} < \smallvalue}
        \absoluteValue{
          \frac{\partial^2}{\partial\distributionIndex_{(\dimensionIndexOne)} \partial\distributionIndex_{(\dimensionIndexTwo)}}
          \log \jointDist_{\noiselessSignal, \noisySignal | \distributionIndex}(\noiselessSignal, \noisySignal)
        }^2
    &<
    \infty \\
    \Expectation_{\jointDist_{\noiselessSignal, \noisySignal | \distributionIndex_0}}
      \absoluteValue{
        \frac{\partial}{\partial\distributionIndex_{(\dimensionIndexOne)}}
        \log \jointDist_{\noiselessSignal, \noisySignal | \distributionIndex_0}(\noiselessSignal, \noisySignal)
      }^2
    &<
    \infty \\
    \Expectation_{\jointDist_{\noisySignal | \distributionIndex_0}}
      \sup_{\distributionIndex: \euclidnorm{\distributionIndex-\distributionIndex_0} < \smallvalue}
        \absoluteValue{
          \frac{\partial^2}{\partial\distributionIndex_{(\dimensionIndexOne)} \partial\distributionIndex_{(\dimensionIndexTwo)}}
          \log \jointDist_{\noisySignal | \distributionIndex}(\noisySignal)
        }^2
    &<
    \infty \\
    \Expectation_{\jointDist_{\noisySignal | \distributionIndex_0}}
      \absoluteValue{
        \frac{\partial}{\partial\distributionIndex_{(\dimensionIndexOne)}}
        \log \jointDist_{\noisySignal | \distributionIndex_0}(\noisySignal)
      }^2
    &<
    \infty,
  \end{align*}
  where $\distributionIndex_{(\dimensionIndexOne)}$ and $\distributionIndex_{(\dimensionIndexTwo)}$ denote the $\dimensionIndexOne$-th, respectively $\dimensionIndexTwo$-th component of $\distributionIndex$.
  \item\label{item:cb-smoothness-kl} $\kldiv{\jointDist_{\noiselessSignal, \noisySignal | \distributionIndex_0}}{\jointDist_{\noiselessSignal, \noisySignal | \distributionIndex}}$ and $\kldiv{\jointDist_{\noisySignal | \distributionIndex_0}}{\jointDist_{\noisySignal | \distributionIndex}}$ are twice continuously differentiable at $\distributionIndex_0$ as functions of $\distributionIndex$, and they have a positive definite Hessian at $\distributionIndex_0$.
  \item\label{item:cb-fisher-regularity} For all $\dimensionIndexOne, \dimensionIndexTwo \in \{1, \dots, \cbdimension\}$, we have
  \begin{align*}
    \frac{\partial^2}{\partial\distributionIndex_{(\dimensionIndexOne)} \partial\distributionIndex_{(\dimensionIndexTwo)}}
    \int_{\noiselessSignalAlphabet \times \noisySignalAlphabet}
      \jointDist_{\noiselessSignal, \noisySignal | \distributionIndex_0}
        (\noiselessSignalValue, \noisySignalValue)
      d\noiselessSignalValue d\noisySignalValue
    &=
    \int_{\noiselessSignalAlphabet \times \noisySignalAlphabet}
      \frac{\partial^2}{\partial\distributionIndex_{(\dimensionIndexOne)} \partial\distributionIndex_{(\dimensionIndexTwo)}}
      \jointDist_{\noiselessSignal, \noisySignal | \distributionIndex_0}
        (\noiselessSignalValue, \noisySignalValue)
      d\noiselessSignalValue d\noisySignalValue
    \\
    \frac{\partial^2}{\partial\distributionIndex_{(\dimensionIndexOne)} \partial\distributionIndex_{(\dimensionIndexTwo)}}
    \int_{\noisySignalAlphabet}
      \jointDist_{\noisySignal | \distributionIndex_0}
        (\noisySignalValue)
      d\noisySignalValue
    &=
    \int_{\noisySignalAlphabet}
      \frac{\partial^2}{\partial\distributionIndex_{(\dimensionIndexOne)} \partial\distributionIndex_{(\dimensionIndexTwo)}}
      \jointDist_{\noisySignal | \distributionIndex_0}
        (\noisySignalValue)
      d\noisySignalValue.
  \end{align*}
  \item\label{item:cb-positive-density} $\distributionIndexWeight$ is described by a continuous density and $\distributionIndexWeight(\distributionIndex_0) > 0$.
  \item\label{item:cb-soundness} For every sequence $(\distributionIndex_\generalNatural)_{\generalNatural \in \naturals}$ of elements of $\distributionIndexDomain$, we have
  \[
    \jointDist_{\noiselessSignal, \noisySignal | \distributionIndex_\generalNatural}
    \xrightarrow[\generalNatural \rightarrow \infty]{}
    \jointDist_{\noiselessSignal, \noisySignal | \distributionIndex_0}
    \Leftrightarrow
    \distributionIndex_\generalNatural
    \xrightarrow[\generalNatural \rightarrow \infty]{}
    \distributionIndex_0
    \Leftrightarrow
    \jointDist_{\noisySignal | \distributionIndex_\generalNatural}
    \xrightarrow[\generalNatural \rightarrow \infty]{}
    \jointDist_{\noisySignal | \distributionIndex_0},
  \]
  where the convergence of probability measures is in the weak sense.
\end{enumerate}
Then, we have
\begin{align*}
\lim_{\seqLength \rightarrow \infty}
\regretPerComp
&\leq
\frac{\maxLoss}{2}
\sqrt{
  \log \frac{\det \fisherInformation_{\noiselessSignal, \noisySignal}(\distributionIndex_0)}
            {\det \fisherInformation_{\noisySignal}(\distributionIndex_0)}
}
&\text{if the loss is additive with $\lossFunc$ bounded by $\maxLoss$,}
\\
\regretPerComp
&=
\frac{1}{2\seqLength} \log \frac{\det \fisherInformation_{\noiselessSignal, \noisySignal}(\distributionIndex_0)}
                                {\det \fisherInformation_{\noisySignal}(\distributionIndex_0)}
+
\landauo\left(\frac{1}{\seqLength}\right)
&\text{if $\lossFunc$ is the self-information loss,}
\end{align*}
where $\fisherInformation_{\noiselessSignal, \noisySignal}(\distributionIndex_0)$ and $\fisherInformation_{\noisySignal}(\distributionIndex_0)$ denote the Fisher information matrix of the families $(\jointDist_{\noiselessSignal, \noisySignal | \distributionIndex})_{\distributionIndex \in \distributionIndexDomain}$ and $(\jointDist_{\noisySignal | \distributionIndex})_{\distributionIndex \in \distributionIndexDomain}$, respectively, at $\distributionIndex_0$, and $\landauo$ is the Landau symbol denoting asymptotic behavior of terms as $\seqLength \rightarrow \infty$.
\end{theorem}
The proof is a straightforward application of Theorem~\ref{theorem:oneshot} and the results in \cite{clarke2002information} and is given in Appendix~\ref{appendix:iid-proof-continuous}. Conditions \ref{item:cb-smoothness-density} and \ref{item:cb-smoothness-kl} are smoothness assumptions for the parametric families and their Kullback-Leibler divergences. Condition \ref{item:cb-fisher-regularity} is a regularity condition for the Fisher information which assures that the Fisher information can be written in terms of the Hessian of the log likelihood function. Conditions \ref{item:cb-positive-density} and \ref{item:cb-soundness}, on the other hand, can be seen as continuous analogs of the conditions that $\distributionIndexWeight$ is supported on $\distributionIndex_0$ and that \eqref{eq:consistency} holds. One important example given in \cite{clarke2002information} (after Theorem 2.2) is that these assumptions are satisfied for the case of finite-dimensional exponential families. The regret bounds of Theorem~\ref{theorem:iid-denoising-continuous} assure universality of the denoising strategy only for the case of self-information loss in the sense of Definition~\ref{def:universal-denoising}-\ref{item:universal-denoising-perComp}. For the case of bounded loss, while the theorem does not assure a vanishing regret, it does give an upper bound for the asymptotic average regret per component which can still be useful in case the Fisher information can be uniformly bounded for the entire parametric family of distributions.

One remarkable property of Theorem~\ref{theorem:iid-denoising-continuous} is that the result on self-information loss is given exactly (up to higher-order terms) rather than just as an upper bound. Multiplying both sides of the equation with $\seqLength$, we get
\[
\regret
=
\frac{1}{2} \log \frac{\det \fisherInformation_{\noiselessSignal, \noisySignal}(\distributionIndex_0)}
                      {\det \fisherInformation_{\noisySignal}(\distributionIndex_0)}
+
\landauo\left(1\right),
\]
which shows that the total regret does not vanish as $\seqLength \rightarrow \infty$. This is in contrast with the result of Theorem~\ref{theorem:iid-denoising} for the case of countable $\distributionIndexDomain$ and means that (at least as long as the fairly standard regularity conditions of Theorem~\ref{theorem:iid-denoising-continuous} are satisfied), there seems to be a fundamental difference between the cases of countable and continuous $\distributionIndexDomain$ when it comes to the behavior of the denoising strategy defined in \eqref{eq:mismatched-denoising} and \eqref{eq:universal-dist}.

\subsection{Necessity of condition \eqref{eq:consistency}}
\label{sec:iid-necessary}
In this subsection, we give an example of a parametric family which does not satisfy \eqref{eq:consistency} and for which it is not possible to find a universal denoising strategy. This shows that condition \eqref{eq:consistency} is necessary for universal denoising regardless of the denoising strategy employed.
\begin{example}
\label{example:consistency-necessary}
Let $\noiselessSignalAlphabet = \noisySignalAlphabet := \{0,1\}$ be binary alphabets, and let $\distributionIndexDomain := \{1,2\}$. Furthermore, let $\jointDist_{\noiselessSignal,\noisySignal | \distributionIndex} := \jointDist_{\noisySignal | \distributionIndex} \jointDist_{\noiselessSignal | \noisySignal, \distributionIndex}$ where
\begin{align*}
  \jointDist_{\noisySignal | 1} = \jointDist_{\noisySignal | 2} &:= \bernoulli{\frac{1}{2}} \\
  \jointDist_{\noiselessSignal | \noisySignal, 1} &:= \diracDist{\noisySignal} \\
  \jointDist_{\noiselessSignal | \noisySignal, 2} &:= \diracDist{1-\noisySignal}.
\end{align*}
As a loss function, we consider both the classification loss, i.e., $\estSignalAlphabet = \{0,1\}$ and $\lossFunc(\estSignalValue,\noiselessSignalValue) = \indicator{\estSignalValue \neq \noiselessSignalValue}$ (which is an additive loss bounded by $1$), and the self-information loss.
\end{example}
Clearly, with knowledge of $\distributionIndex$, it is possible to achieve zero loss in both cases. Specifically, consider
\[
\estSignalOpt_\seqIndex :=
\begin{cases}
\noisySignal_\seqIndex, &\text{$\lossFunc$ is classification loss, $\distributionIndex=1$}\\
1-\noisySignal_\seqIndex, &\text{$\lossFunc$ is classification loss, $\distributionIndex=2$}\\
\diracDist{\noisySignal_\seqIndex}, &\text{$\lossFunc$ is self-information loss, $\distributionIndex=1$}\\
\diracDist{1-\noisySignal_\seqIndex}, &\text{$\lossFunc$ is self-information loss, $\distributionIndex=2$.}
\end{cases}
\]
This means that the Bayes envelope is $0$ and the denoiser's regret is equal to its loss. Assume that the denoising strategy that yields $\estSignal^\seqLength$ as a function of $\noisySignal^\seqLength$ is universal in the weaker sense of Definition~\ref{def:universal-denoising}-\ref{item:universal-denoising-perComp}. It clearly must satisfy
\[
\Expectation_{\jointDist_{\noiselessSignal^\seqLength, \noisySignal^\seqLength | \distributionIndex}} \sum_{\seqIndex=1}^\infty \lossFunc(\estSignal_\seqIndex,\noiselessSignal_\seqIndex) < \infty
\]
for both values of $\distributionIndex$, which implies in particular $\lim_{\seqIndex \rightarrow \infty} \Expectation_{\jointDist_{\noiselessSignal^\seqLength, \noisySignal^\seqLength | \distributionIndex}} \lossFunc(\estSignal_\seqIndex,\noiselessSignal_\seqIndex) = 0$. Note also that the definition of $\jointDist_{\noiselessSignal|\noisySignal,\distributionIndex}$ in Example~\ref{example:consistency-necessary} implies that $\jointDist_{\noiselessSignal,\noisySignal | 1} = \jointDist_{1-\noiselessSignal,\noisySignal | 2}$. Consequently, we have
\begin{equation}
\label{eq:consistency-necessary-joint-conclusion}
\lim_{\seqIndex \rightarrow \infty}
  \Expectation_{\jointDist_{\noiselessSignal^\seqLength, \noisySignal^\seqLength | 2}}
    \lossFunc(\estSignal_\seqIndex,\noiselessSignal_\seqIndex)
=
0
=
\lim_{\seqIndex \rightarrow \infty}
  \Expectation_{\jointDist_{\noiselessSignal^\seqLength, \noisySignal^\seqLength | 1}}
    \lossFunc(\estSignal_\seqIndex,\noiselessSignal_\seqIndex)
=
\lim_{\seqIndex \rightarrow \infty}
  \Expectation_{\jointDist_{\noiselessSignal^\seqLength, \noisySignal^\seqLength | 2}}
    \lossFunc(\estSignal_\seqIndex,1-\noiselessSignal_\seqIndex)
\end{equation}

\paragraph{Classification loss}
In this case we can write
\[
\Expectation_{\jointDist_{\noiselessSignal^\seqLength, \noisySignal^\seqLength | \distributionIndex}}
  \lossFunc(\estSignal_\seqIndex,\noiselessSignal_\seqIndex)
=
\jointDist_{\noiselessSignal^\seqLength, \noisySignal^\seqLength | \distributionIndex}\left(
  \estSignal_\seqIndex \neq \noiselessSignal_\seqIndex
\right).
\]
Substituting this in \eqref{eq:consistency-necessary-joint-conclusion}, we obtain
\[
\lim_{\seqIndex \rightarrow \infty}
  \jointDist_{\noiselessSignal^\seqLength, \noisySignal^\seqLength | 2}\left(
    \estSignal_\seqIndex \neq \noiselessSignal_\seqIndex
  \right)
=
0
=
\lim_{\seqIndex \rightarrow \infty}
  \jointDist_{\noiselessSignal^\seqLength, \noisySignal^\seqLength | 2}\left(
    \estSignal_\seqIndex = \noiselessSignal_\seqIndex
  \right).
\]
However,
$
\lim_{\seqIndex \rightarrow \infty}
  \jointDist_{\noiselessSignal^\seqLength, \noisySignal^\seqLength | 2}\left(
    \estSignal_\seqIndex \neq \noiselessSignal_\seqIndex
  \right)
=
0
$
implies
\[
\lim_{\seqIndex \rightarrow \infty}
  \jointDist_{\noiselessSignal^\seqLength, \noisySignal^\seqLength | 2}\left(
    \estSignal_\seqIndex = \noiselessSignal_\seqIndex
  \right)
=
\lim_{\seqIndex \rightarrow \infty}\left(
  1
  -
  \jointDist_{\noiselessSignal^\seqLength, \noisySignal^\seqLength | 2}\left(
    \estSignal_\seqIndex \neq \noiselessSignal_\seqIndex
  \right)
\right)
=
1,
\]
which contradicts the original assumption that the denoising strategy is universal.

\paragraph{Self-information loss}
In this case we can write
\[
\Expectation_{\jointDist_{\noiselessSignal^\seqLength, \noisySignal^\seqLength | \distributionIndex}}
  \lossFunc(\estSignal_\seqIndex,\noiselessSignal_\seqIndex)
=
-
\Expectation_{\jointDist_{\noiselessSignal^\seqLength, \noisySignal^\seqLength | \distributionIndex}}
  \log \estSignal_\seqIndex(\noiselessSignal_\seqIndex).
\]
Substituting this in \eqref{eq:consistency-necessary-joint-conclusion}, we obtain
\[
\lim_{\seqIndex \rightarrow \infty}
  \Expectation_{\jointDist_{\noiselessSignal^\seqLength, \noisySignal^\seqLength | 2}}
    \log \estSignal_\seqIndex(\noiselessSignal_\seqIndex)
=
0
=
\lim_{\seqIndex \rightarrow \infty}
  \Expectation_{\jointDist_{\noiselessSignal^\seqLength, \noisySignal^\seqLength | 2}}
    \log \estSignal_\seqIndex(1-\noiselessSignal_\seqIndex).
\]
Since $0$ is the maximum possible value of $\log \estSignal_\seqIndex(1-\noiselessSignal_\seqIndex)$, this implies that we $\jointDist_{\noiselessSignal^\seqLength, \noisySignal^\seqLength | 2}$-almost surely have
\[
\lim_{\seqIndex \rightarrow \infty}
  \log \estSignal_\seqIndex(\noiselessSignal_\seqIndex)
=
0
=
\lim_{\seqIndex \rightarrow \infty}
  \log \estSignal_\seqIndex(1-\noiselessSignal_\seqIndex).
\]
Due to the continuity of the logarithm, we can equivalently write
\[
\lim_{\seqIndex \rightarrow \infty}
  \estSignal_\seqIndex(\noiselessSignal_\seqIndex)
=
1
=
\lim_{\seqIndex \rightarrow \infty}
  \estSignal_\seqIndex(1-\noiselessSignal_\seqIndex).
\]
But since $\estSignal_\seqIndex(1-\noiselessSignal_\seqIndex) = 1 - \estSignal_\seqIndex(\noiselessSignal_\seqIndex)$, we have again arrived at a contradiction of the original assumption that the denoising strategy is universal.

\section{The Plug-in Approach}
\label{sec:plug-in}
It is not necessarily intuitively clear why the denoising strategy given by \eqref{eq:mismatched-denoising} and \eqref{eq:universal-dist} is required. A simpler and more straightforward strategy could be to determine a maximum likelihood estimate $\distributionIndex_\mathrm{ML}$ of $\distributionIndex_0$ from the observation $\noisySignal^\seqLength$ and then estimate the $\estSignal^\seqLength$ which has the lowest expected loss under $\jointDist_{\noiselessSignal | \noisySignal, \distributionIndex_\mathrm{ML}}^\seqLength$. This so-called \emph{plug-in approach} is known to be in general suboptimal in the context of universal prediction~\cite{merhav1998universal}, and in this section we show that this is indeed also the case for universal denoising. Intuitively, there might be $\distributionIndex \in \distributionIndexDomain$ for which $\jointDist_{\noisySignal | \distributionIndex}$ and $\jointDist_{\noisySignal | \distributionIndex_0}$ are very similar while $\jointDist_{\noiselessSignal | \noisySignal, \distributionIndex}$ and $\jointDist_{\noiselessSignal | \noisySignal, \distributionIndex_0}$ differ dramatically. While the plug-in approach always has to commit to one hypothesis on what $\distributionIndex_0$ is, the mixture approach given by \eqref{eq:mismatched-denoising} and \eqref{eq:universal-dist} is more flexible in that it always at least partially considers the true $\jointDist_{\noiselessSignal | \noisySignal, \distributionIndex_0}$. We make this intuition more rigorous in the following example.

\begin{example}
\label{example:plugin-bad}
Let $\noiselessSignalAlphabet := \noisySignalAlphabet := \estSignalAlphabet := \{0,1\}$, consider the classification loss function $\lossFunc: (\estSignalValue,\noiselessSignalValue) \mapsto \indicator{\estSignalValue\neq\noiselessSignalValue}$, let $\distributionIndexDomain := \{0, 1, 2, \dots\}$ with $\distributionIndex_0 := 0$, and define $\jointDist_{\noiselessSignal, \noisySignal | \distributionIndex}$ via
\begin{align*}
\jointDist_{\noiselessSignal | 0} &:= \frac{1}{2} \diracDist{0} + \frac{1}{2} \diracDist{1},&
\jointDist_{\noisySignal|\noiselessSignal, 0} &:= \diracDist{\noiselessSignal}
\\
\text{for } \distributionIndex \in \distributionIndexDomain \setminus \{0\}:~
\jointDist_{\noiselessSignal | \distributionIndex}
&:=
\left( \frac{1}{2} + \frac{1}{2+\distributionIndex}(-1)^{\distributionIndex}\right) \diracDist{0}
+
\left( \frac{1}{2} - \frac{1}{2+\distributionIndex}(-1)^{\distributionIndex}\right) \diracDist{1},&
\jointDist_{\noisySignal|\noiselessSignal, \distributionIndex} &:= \diracDist{1-\noiselessSignal}
\end{align*}
\end{example}

\paragraph{Plug-in approach} The learner's strategy we analyze is as follows: based on the noisy observations $\noisySignal_1, \dots, \noisySignal^\seqLength$, choose the maximum likelihood parameter
\[
\distributionIndex_\mathrm{ML}
:=
\argmax_{\distributionIndex \in \distributionIndexDomain}
  \jointDist_{\noisySignal | \distributionIndex}(\noisySignal_1)
  \cdots
  \jointDist_{\noisySignal | \distributionIndex}(\noisySignal_\seqLength)
\]
and use denoising strategy \eqref{eq:mismatched-denoising} with $\universalDist_{\noiselessSignal^\seqLength, \noisySignal^\seqLength} := \jointDist_{\noiselessSignal, \noisySignal | \distributionIndex_\mathrm{ML}}^\seqLength$. Clearly, the total regret will be $0$ if $\distributionIndex_\mathrm{ML} = 0$ and $\seqLength$ otherwise.

\begin{lemma}
\label{lemma:binomial-ml}
$\distributionIndex_\mathrm{ML} = 0$ iff $\countingFunction(\noisySignal^\seqLength,0) = \frac{\seqLength}{2}$.
\end{lemma}
\begin{proof}
Let
\[
\binomLikelihood(\binomParameter; \binomTrials, \binomSuccesses)
:=
\log\left(
  \binom{\binomTrials}{\binomSuccesses}
  \binomParameter^{\binomSuccesses}
  (1-\binomParameter)^{\binomTrials-\binomSuccesses}
\right)
=
\log \binom{\binomTrials}{\binomSuccesses}
+
\binomSuccesses \log \binomParameter
+
(\binomTrials-\binomSuccesses) \log (1-\binomParameter)
\]
denote the binomial log-likelihood of parameter $\binomParameter$ given $\binomTrials$ trials and $\binomSuccesses$ successes. We compute its derivative as
\[
\binomLikelihood'(\binomParameter; \binomTrials, \binomSuccesses)
:=
\frac{\partial}
     {\partial \binomParameter}
\binomLikelihood(\binomParameter; \binomTrials, \binomSuccesses)
=
\frac{\binomSuccesses}{\binomParameter}
-
\frac{\binomTrials-\binomSuccesses}{1-\binomParameter}.
\]
Clearly, $\binomLikelihood(\binomParameter; \binomTrials, \binomSuccesses)$ is a continuous function of $\binomParameter$, and it can be checked that
\begin{align*}
\binomLikelihood'(\binomParameter; \binomTrials, \binomSuccesses) &> 0 \text{ (i.e., $\binomLikelihood$ is increasing) }& \text{ for } \binomParameter &< \binomSuccesses/\binomTrials
\\
\binomLikelihood'(\binomParameter; \binomTrials, \binomSuccesses) &= 0 \text{ (i.e., $\binomLikelihood$ takes its maximum value) }& \text{ for } \binomParameter &= \binomSuccesses/\binomTrials
\\
\binomLikelihood'(\binomParameter; \binomTrials, \binomSuccesses) &< 0 \text{ (i.e., $\binomLikelihood$ is decreasing) }& \text{ for } \binomParameter &> \binomSuccesses/\binomTrials.
\end{align*}
From this, it is immediately clear that if $\countingFunction(\noisySignal^\seqLength,0) = \frac{\seqLength}{2}$, $\binomLikelihood(\binomParameter; \seqLength, \seqLength/2)$ takes its maximum value at $\binomParameter = 1/2$ and therefore $\distributionIndex_\mathrm{ML} = 0$. If, on the other hand, $\countingFunction(\noisySignal^\seqLength,0) \neq \frac{\seqLength}{2}$, then $\binomLikelihood(\binomParameter; \seqLength, \countingFunction(\noisySignal^\seqLength,0))$ takes its maximum value at $\binomParameter=\binomParameter_{\max}:=\countingFunction(\noisySignal^\seqLength,0)/\seqLength \neq 1/2$. Clearly, there is $\distributionIndex \in \distributionIndexDomain \setminus \{0\}$ such that
\[
\frac{1}{2} + \frac{1}{2+\distributionIndex}(-1)^{\distributionIndex} \in \left(\frac{1}{2}, \binomParameter_{\max}\right)
\text{ or }
\frac{1}{2} + \frac{1}{2+\distributionIndex}(-1)^{\distributionIndex} \in \left(\binomParameter_{\max}, \frac{1}{2}\right)
\]
(only one of these intervals will be nonempty). From the established monotonicity behavior of $\binomLikelihood$, it is clear that $\binomLikelihood(1/2; \seqLength, \countingFunction(\noisySignal^\seqLength,0)) < \binomLikelihood(1/2; \seqLength, \countingFunction(\noisySignal^\seqLength,\distributionIndex))$ and hence $\distributionIndex_\mathrm{ML} \neq 0$.
\end{proof}
Lemma~\ref{lemma:binomial-ml} establishes that $\regret = \seqLength$ whenever $\seqLength$ is odd. For even $\seqLength$, we can compute
\begin{align}
\nonumber
\regret
&=
\jointDist_{\noiselessSignal, \noisySignal | \distributionIndex_0}\left(\distributionIndex_\mathrm{ML} \neq 0\right) \cdot \seqLength
\\
\nonumber
\overset{(a)}&{=}
\jointDist_{\noiselessSignal, \noisySignal | \distributionIndex_0}\left(\countingFunction(\noisySignal^\seqLength, 0) \neq \frac{1}{2}\right) \cdot \seqLength
\\
\nonumber
\overset{(b)}&{=}
\left(
  1
  -
  \binom{\seqLength}{\frac{\seqLength}{2}}
  2^{-\seqLength}
\right)
\cdot
\seqLength
\\
\nonumber
&=
\left(
  1
  -
  \frac{\seqLength!}{\left(\frac{\seqLength}{2}\right)!^2}
2^{-\seqLength}
\right)
\cdot
\seqLength
\\
\nonumber
\overset{(c)}&{>}
\left(
  1
  -
  \frac{\sqrt{2\pi} \seqLength^{\seqLength+\frac{1}{2}} \exp(-\seqLength+1)}
       {
         \left(
           \sqrt{2\pi}
           \left(\frac{\seqLength}{2}\right)^{\frac{\seqLength}{2}+\frac{1}{2}}
           \exp\left(-\frac{\seqLength}{2}\right)
         \right)^2
       }
2^{-\seqLength}
\right)
\cdot
\seqLength
\\
\nonumber
&=
\left(
  1
  -
  \frac{\eulersNumber}{\sqrt{\pi}}
  \cdot
  \frac{\left(
        \frac{\seqLength}{\frac{\seqLength}{2}}
        \right)^{\seqLength+\frac{1}{2}}
      }
      {\sqrt{\seqLength}}
2^{-\seqLength}
\right)
\cdot
\seqLength
\\
\label{eq:plugin-lower-bound}
&=
\seqLength
-
\eulersNumber
\sqrt{
  \frac{2\seqLength}{\pi}
},
\end{align}
where in step (a) we have used Lemma~\ref{lemma:binomial-ml}, in step (b) we have substituted the binomial p.m.f., and in step (c) we have used
\[
\sqrt{2\pi}
\generalNatural^{\generalNatural+\frac{1}{2}}
\exp(-\generalNatural)
<
\generalNatural!
<
\sqrt{2\pi}
\generalNatural^{\generalNatural+\frac{1}{2}}
\exp(-\generalNatural+1)
\]
which is a slightly weaker form of the first two equations in~\cite{robbins1955remark}. Hence, in both cases the total regret will grow linearly with $\seqLength$.

\paragraph{Mixture approach}
We analyze the learner's strategy given by \eqref{eq:mismatched-denoising} and \eqref{eq:universal-dist}. Defining the function
\[
\exampleMixFunc(\distributionIndex,\noisySignalValue^\seqLength)
:=
\left(\frac{1}{2} + \frac{1}{2+\distributionIndex}(-1)^{\distributionIndex}\right)^{\countingFunction(\noisySignalValue^\seqLength,0)}
\left(\frac{1}{2} - \frac{1}{2+\distributionIndex}(-1)^{\distributionIndex}\right)^{\countingFunction(\noisySignalValue^\seqLength,1)},
\]
we have
\begin{align*}
\universalDist_{\noiselessSignal^\seqLength, \noisySignal^\seqLength}(\noiselessSignalValue^\seqLength,\noisySignalValue^\seqLength)
&=
\distributionIndexWeight(\distributionIndex_0)
2^{-\seqLength}
\indicator{\noiselessSignalValue^\seqLength=\noisySignalValue^\seqLength}
+
\sum_{\distributionIndex=1}^\infty
  \distributionIndexWeight(\distributionIndex)
  \exampleMixFunc(\distributionIndex,\noisySignalValue^\seqLength)
  \indicator{\noiselessSignalValue^\seqLength=(1-\noisySignalValue_1,\dots,1-\noisySignalValue_\seqLength)},
\\
\universalDist_{\noisySignal^\seqLength}(\noisySignalValue^\seqLength)
&=
\sum_{\noiselessSignalValue^\seqLength \in \{0,1\}^\seqLength}
\universalDist_{\noiselessSignal^\seqLength, \noisySignal^\seqLength}(\noiselessSignalValue^\seqLength,\noisySignalValue^\seqLength)
\\
&=
\distributionIndexWeight(\distributionIndex_0)
2^{-\seqLength}
+
\sum_{\distributionIndex=1}^\infty
  \distributionIndexWeight(\distributionIndex)
  \exampleMixFunc(\distributionIndex,\noisySignalValue^\seqLength)
\\
\universalDist_{\noiselessSignal^\seqLength|\noisySignal^\seqLength}(\noiselessSignalValue^\seqLength|\noisySignalValue^\seqLength)
&=
\frac{\universalDist_{\noiselessSignal^\seqLength, \noisySignal^\seqLength}(\noiselessSignalValue^\seqLength,\noisySignalValue^\seqLength)}
     {\universalDist_{\noisySignal^\seqLength}(\noisySignalValue^\seqLength)}
\\
&=
\begin{multlined}[t]
\indicator{\noiselessSignalValue^\seqLength=\noisySignalValue^\seqLength}
\left(
  1
  +
  2^\seqLength
  \sum_{\distributionIndex=1}^\infty
    \frac{\distributionIndexWeight(\distributionIndex)}{\distributionIndexWeight(\distributionIndex_0)}
    \exampleMixFunc(\distributionIndex,\noisySignalValue^\seqLength)
\right)^{-1}
\\
+
\indicator{\noiselessSignalValue^\seqLength=(1-\noisySignalValue_1,\dots,1-\noisySignalValue_\seqLength)}
\left(
  1
  +
  \left(
    2^\seqLength
    \sum_{\distributionIndex=1}^\infty
      \frac{\distributionIndexWeight(\distributionIndex)}{\distributionIndexWeight(\distributionIndex_0)}
      \exampleMixFunc(\distributionIndex,\noisySignalValue^\seqLength)
  \right)^{-1}
\right)^{-1}.
\end{multlined}
\end{align*}

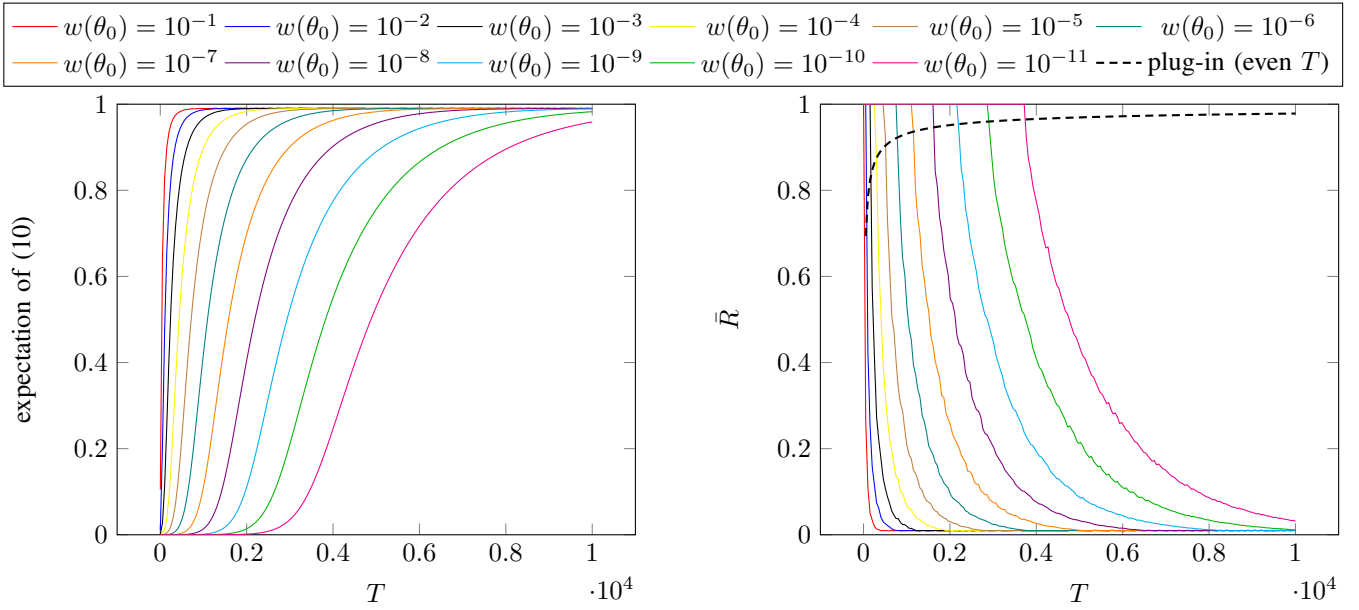
\begin{figure}
\ref*{plugin-legend}

\begin{subfigure}[t]{.45\textwidth}
\begin{tikzpicture}
\begin{axis}[
  ymin=0,
  ymax=1,
  xlabel={$\seqLength$},
  ylabel={expectation of (\ref*{eq:posterior-nonflip-probability})},
  cycle list name=color list,
]
  \addplot table[x=seq_len, y=nonflip_0, col sep=comma] {plugin-example.csv};
  \addplot table[x=seq_len, y=nonflip_1, col sep=comma] {plugin-example.csv};
  \addplot table[x=seq_len, y=nonflip_2, col sep=comma] {plugin-example.csv};
  \addplot table[x=seq_len, y=nonflip_3, col sep=comma] {plugin-example.csv};
  \addplot table[x=seq_len, y=nonflip_4, col sep=comma] {plugin-example.csv};
  \addplot table[x=seq_len, y=nonflip_5, col sep=comma] {plugin-example.csv};
  \addplot table[x=seq_len, y=nonflip_6, col sep=comma] {plugin-example.csv};
  \addplot table[x=seq_len, y=nonflip_7, col sep=comma] {plugin-example.csv};
  \addplot table[x=seq_len, y=nonflip_8, col sep=comma] {plugin-example.csv};
  \addplot table[x=seq_len, y=nonflip_9, col sep=comma] {plugin-example.csv};
  \addplot table[x=seq_len, y=nonflip_10, col sep=comma] {plugin-example.csv};
\end{axis}
\end{tikzpicture}
\caption{Expected non-flipping probability of $\universalDist_{\noiselessSignal^\seqLength|\noisySignal^\seqLength}$ in the mixture approach.}
\label{fig:plugin-nonflip}
\end{subfigure}
\hspace{.05\textwidth}
\begin{subfigure}[t]{.45\textwidth}
\begin{tikzpicture}
\begin{axis}[
  ymin=0,
  ymax=1,
  xlabel={$\seqLength$},
  ylabel={$\regretPerComp$},
  legend to name={plugin-legend},
  legend columns=6,
  cycle list name=color list
]
  \addplot table[x=seq_len, y=regret_0, col sep=comma] {plugin-example.csv};
  \addplot table[x=seq_len, y=regret_1, col sep=comma] {plugin-example.csv};
  \addplot table[x=seq_len, y=regret_2, col sep=comma] {plugin-example.csv};
  \addplot table[x=seq_len, y=regret_3, col sep=comma] {plugin-example.csv};
  \addplot table[x=seq_len, y=regret_4, col sep=comma] {plugin-example.csv};
  \addplot table[x=seq_len, y=regret_5, col sep=comma] {plugin-example.csv};
  \addplot table[x=seq_len, y=regret_6, col sep=comma] {plugin-example.csv};
  \addplot table[x=seq_len, y=regret_7, col sep=comma] {plugin-example.csv};
  \addplot table[x=seq_len, y=regret_8, col sep=comma] {plugin-example.csv};
  \addplot table[x=seq_len, y=regret_9, col sep=comma] {plugin-example.csv};
  \addplot table[x=seq_len, y=regret_10, col sep=comma] {plugin-example.csv};
  \addplot[domain=0:1e4,samples=200,style=densely dashed,thick] {1-e*sqrt(2/(pi*x))};
  \legend{$\distributionIndexWeight(\distributionIndex_0) = 10^{-1}$,
          $\distributionIndexWeight(\distributionIndex_0) = 10^{-2}$,
          $\distributionIndexWeight(\distributionIndex_0) = 10^{-3}$,
          $\distributionIndexWeight(\distributionIndex_0) = 10^{-4}$,
          $\distributionIndexWeight(\distributionIndex_0) = 10^{-5}$,
          $\distributionIndexWeight(\distributionIndex_0) = 10^{-6}$,
          $\distributionIndexWeight(\distributionIndex_0) = 10^{-7}$,
          $\distributionIndexWeight(\distributionIndex_0) = 10^{-8}$,
          $\distributionIndexWeight(\distributionIndex_0) = 10^{-9}$,
          $\distributionIndexWeight(\distributionIndex_0) = 10^{-10}$,
          $\distributionIndexWeight(\distributionIndex_0) = 10^{-11}$,
          plug-in (even $\seqLength$)}
\end{axis}
\end{tikzpicture}
\caption{Expected regret of the learner per component.}
\label{fig:plugin-regret}
\end{subfigure}
\caption{Numerical comparison of the performance of the plug-in approach with the mixture approach.}
\label{fig:plugin-plots}
\end{figure}

This means that the channel from $\noiselessSignal^\seqLength$ to $\noisySignal^\seqLength$ under $\universalDist_{\noiselessSignal^\seqLength, \noisySignal^\seqLength}(\noiselessSignalValue^\seqLength,\noisySignalValue^\seqLength)$ returns its input with probability
\begin{equation}
\label{eq:posterior-nonflip-probability}
\left(
  1
  +
  2^\seqLength
  \sum_{\distributionIndex=1}^\infty
    \frac{\distributionIndexWeight(\distributionIndex)}{\distributionIndexWeight(\distributionIndex_0)}
    \exampleMixFunc(\distributionIndex,\noisySignalValue^\seqLength)
\right)^{-1}
\end{equation}
and otherwise flips all of its input bits. Under the classification loss, the learner will return the estimate $\estSignal^\seqLength$ that is most likely to equal $\noiselessSignal^\seqLength$ given the observation $\noisySignal^\seqLength$ under $\universalDist_{\noiselessSignal^\seqLength, \noisySignal^\seqLength}(\noiselessSignalValue^\seqLength,\noisySignalValue^\seqLength)$. This means that the regret is $0$ whenever the quantity \eqref{eq:posterior-nonflip-probability} is greater than $1/2$ and the regret is $\seqLength$ whenever \eqref{eq:posterior-nonflip-probability} is less than $1/2$.


It is clear from Theorem~\ref{theorem:iid-denoising} that the regret incurred with the mixture approach vanishes as $\seqLength \rightarrow \infty$. In order to get a better idea of what this looks like numerically and what the conditional mixture distribution $\universalDist_{\noiselessSignal^\seqLength|\noisySignal^\seqLength}$ looks like on average, we show the expectation of the probability \eqref{eq:posterior-nonflip-probability} along with the expected regret the denoiser incurs with the mixture approach compared to the plug-in lower bound \eqref{eq:plugin-lower-bound} in Fig.~\ref{fig:plugin-plots}. Since countably infinite $\distributionIndexDomain$ do not admit uniform $\distributionIndexWeight$, the convergence speed crucially depends on $\distributionIndexWeight(\distributionIndex_0)$. We show plots for various values of $\distributionIndexWeight(\distributionIndex_0)$ (which for realistic scenarios can be expected to be very small), where the probability masses of the other parameters are chosen as
\[
  \distributionIndexWeight(\distributionIndex)
  :=
  2^{-\distributionIndex} (1-\distributionIndexWeight(\distributionIndex_0)).
\]
We can see that even for very small values of $\distributionIndexWeight(\distributionIndex_0)$, the expected regret vanishes in the mixture approach for reasonable sequence lengths whereas it converges to its maximum value $1$ quite quickly with the plug-in approach.

\section{Numerical Example: Denoising of Poisson Signals}
\label{sec:poisson}
Although this work is primarily focused on the theoretical results of Theorems~\ref{theorem:oneshot}, \ref{theorem:iid-denoising}, and \ref{theorem:iid-denoising-continuous}, we show in this section that it is in principle possible to implement the denoising strategy we have analyzed. The main practical difficulty that arises is solving the optimization problem \eqref{eq:mismatched-denoising} with the probability distribution defined in \eqref{eq:universal-dist}. Whether or not this is tractable and how difficult and computationally expensive it is depends heavily on the underlying parametric family $(\jointDist_{\noiselessSignal, \noisySignal | \distributionIndex})_{\distributionIndex \in \distributionIndexDomain}$. Here, we show how it can be done in the case of an example of a denoising problem in which the signals follow a Poisson distribution.

\subsection{Preliminaries}
In this subsection, we establish a lemma that holds with some generality but will later help us in implementing our denoising strategy for the numerical example given in Section~\ref{section:poisson}. The lemma shows that the loss can be estimated with respect to the per-component marginals which makes the computation of \eqref{eq:mismatched-denoising} computationally more tractable.

\begin{lemma}
\label{lemma:simplified-additive}
\emph{(Simplified optimization for additive losses)}
In the case of denoising of (not necessarily i.i.d.) sequences, if the loss is additive, that is
\[
\totalLoss(\estSignalValue^\seqLength, \noiselessSignalValue^\seqLength)
=
\sum_{\seqIndex=1}^\seqLength
  \lossFunc(\estSignalValue_\seqIndex, \noiselessSignalValue_\seqIndex),
\]
and $\noiselessSignalAlphabet$ is countable, then $\estSignal^\seqLength$ solves the optimization \eqref{eq:mismatched-denoising} if and only if
\[
\forall \seqIndex\in\{1,\dots,\seqLength\}:~
\estSignalValue_\seqIndex
\in
\argmax
\Expectation_{\universalDist_{\noiselessSignal_\seqIndex | \noisySignal^\seqLength}}
  \lossFunc(\estSignalValue_\seqIndex,\noiselessSignal_\seqIndex).
\]
\end{lemma}
\begin{proof}
We calculate
\begin{align*}
\Expectation_{\universalDist_{\noiselessSignal^\seqLength | \noisySignal^\seqLength}} \totalLoss(\estSignalValue^\seqLength, \noiselessSignal^\seqLength)
&=
\sum_{\noiselessSignalValue^\seqLength \in \noiselessSignalAlphabet^\seqLength}\left(
  \universalDist_{\noiselessSignal^\seqLength | \noisySignal^\seqLength}(\noiselessSignalValue^\seqLength | \noisySignalValue^\seqLength)
  \sum_{\seqIndex=1}^\seqLength
    \lossFunc(\estSignalValue_\seqIndex, \noiselessSignalValue_\seqIndex)
\right)
\\
&=
\sum_{\seqIndex=1}^\seqLength\left(
  \sum_{\noiselessSignalValue_1 \in \noiselessSignalAlphabet}
  \cdots
  \sum_{\noiselessSignalValue_\seqLength \in \noiselessSignalAlphabet}
    \universalDist_{\noiselessSignal^\seqLength | \noisySignal^\seqLength}(\noiselessSignalValue^\seqLength | \noisySignalValue^\seqLength)
    \lossFunc(\estSignalValue_\seqIndex, \noiselessSignalValue_\seqIndex)
\right)
\\
&=
\sum_{\seqIndex=1}^\seqLength\left(
  \sum_{\noiselessSignalValue_\seqIndex \in \noiselessSignalAlphabet}\left(
    \lossFunc(\estSignalValue_\seqIndex, \noiselessSignalValue_\seqIndex)
    \sum_{\noiselessSignalValue_1 \in \noiselessSignalAlphabet}
    \cdots
    \sum_{\noiselessSignalValue_{\seqIndex-1} \in \noiselessSignalAlphabet}
    \sum_{\noiselessSignalValue_{\seqIndex+1} \in \noiselessSignalAlphabet}
    \cdots
    \sum_{\noiselessSignalValue_\seqLength \in \noiselessSignalAlphabet}
      \universalDist_{\noiselessSignal^\seqLength | \noisySignal^\seqLength}(\noiselessSignalValue^\seqLength | \noisySignalValue^\seqLength)
  \right)
\right)
\\
&=
\sum_{\seqIndex=1}^\seqLength\left(
  \sum_{\noiselessSignalValue_\seqIndex \in \noiselessSignalAlphabet}\left(
    \lossFunc(\estSignalValue_\seqIndex, \noiselessSignalValue_\seqIndex)
    \universalDist_{\noiselessSignal_\seqIndex | \noisySignal^\seqLength}(\noiselessSignalValue_\seqIndex | \noisySignalValue^\seqLength)
  \right)
\right)
\\
&=
\sum_{\seqIndex=1}^\seqLength
  \Expectation_{\universalDist_{\noiselessSignal_\seqIndex | \noisySignal^\seqLength}}
    \lossFunc(\estSignalValue_\seqIndex, \noiselessSignalValue_\seqIndex).
\end{align*}
Since every summand depends on a different entry of $\estSignal^\seqLength$, optimizing the entire sum is equivalent to optimizing every summand.
\end{proof}

\subsection{Denoising of Poisson Signals}
\label{section:poisson}
We consider a denoising problem where the noiseless signal is the number of events generated by an unknown number $\numNoise$ of Poisson sources of known mean $\noiselessIntensity$. The noisy signal is corrupted by additive noise which is the number of events generated by an unknown number $\numNoise$ of Poisson sources of known mean $\noiseIntensity$.
\begin{example}
\label{example:poisson}
\emph{(Denoising of Poisson Signals).}
Consider an instance of Problem~\ref{problem:iid} where $\noiselessSignalAlphabet = \noisySignalAlphabet = \naturals \cup \{0\}$, and $\jointDist_{\noiselessSignal, \noisySignal | \numNoiseless, \numNoise}$ is described as follows:
\begin{itemize}
  \item $\noiselessSignal$ follows a Poisson distribution with mean $\numNoiseless\noiselessIntensity$,
  \item $\additiveNoise$ follows a Poisson distribution with mean $\numNoise\noiseIntensity$,
  \item $\noisySignal = \noiselessSignal + \additiveNoise$.
\end{itemize}
We assume that $\noiselessIntensity, \noiseIntensity \in (0,\infty)$ are known to the denoiser while $\numNoiseless$ and $\numNoise$ are unknown. That is, $\distributionIndexDomain = \naturals^2$, and the parametric family is $(\jointDist_{\noiselessSignal, \noisySignal | \numNoiseless, \numNoise})_{\numNoiseless, \numNoise \in \naturals}$. The true distribution is $\jointDist_{\noiselessSignal, \noisySignal} = \jointDist_{\noiselessSignal, \noisySignal | \trueNumNoiseless, \trueNumNoise}$.
\end{example}
Note that if $\generalNatural$ i.i.d. random variables follow a Poisson distribution with mean $\generalIntensity$, their sum is Poisson distributed with mean $\generalNatural\generalIntensity$. Hence both the noiseless signal and the noise in Example~\ref{example:poisson} can be seen as signals generated by an unknown number of sources with known mean.

\begin{lemma}
The parametric family in Example~\ref{example:poisson} satisfies \eqref{eq:consistency} iff $\noiselessIntensity/\noiseIntensity$ is irrational.
\end{lemma}
\begin{proof}
Two Poisson distributions are identical iff they have the same mean, so clearly $\jointDist_{\noiselessSignal,\noisySignal | \distributionIndex_1} \neq \jointDist_{\noiselessSignal,\noisySignal | \distributionIndex_2}$ for all $\distributionIndex_1 \neq \distributionIndex_2 \in \distributionIndexDomain$. $\noisySignal$ is the sum of two independent Poisson variables, so it follows a Poisson distribution with mean $\numNoise\noiseIntensity + \numNoiseless\noiselessIntensity$. Assume first that there are $\numNoiselessN{1}, \numNoiselessN{2}, \numNoiseN{1}, \numNoiseN{2} \in \naturals$ such that $(\numNoiselessN{1}, \numNoiseN{1}) \neq (\numNoiselessN{2}, \numNoiseN{2})$ and
$
\numNoiseN{1}\noiseIntensity + \numNoiselessN{1}\noiselessIntensity
=
\numNoiseN{2}\noiseIntensity + \numNoiselessN{2}\noiselessIntensity.
$
If $\numNoiselessN{1}=\numNoiselessN{2}$, we can divide both sides by $\noiseIntensity$ and obtain $\numNoiseN{1}=\numNoiseN{2}$, which is a contradiction. Otherwise, we have
\[
\numNoiseN{1}\noiseIntensity + \numNoiselessN{1}\noiselessIntensity
=
\numNoiseN{2}\noiseIntensity + \numNoiselessN{2}\noiselessIntensity
\Leftrightarrow
\frac{\noiselessIntensity}{\noiseIntensity}
=
\frac{\numNoiseN{1}-\numNoiseN{2}}{\numNoiselessN{2}-\numNoiselessN{1}},
\]
so clearly $\noiselessIntensity/\noiseIntensity$ is rational. On the other hand, if $\noiselessIntensity/\noiseIntensity$ is rational, write $\noiselessIntensity/\noiseIntensity = \generalNatural_1/\generalNatural_2$. So for instance choosing $\numNoiseN{1} := \generalNatural_1 + 1, \numNoiseN{2} := 1, \numNoiselessN{1} = 1, \numNoiselessN{2} = \generalNatural_2 + 1$ yields
$
\numNoiseN{1}\noiseIntensity + \numNoiselessN{1}\noiselessIntensity
=
\numNoiseN{2}\noiseIntensity + \numNoiselessN{2}\noiselessIntensity
$,
hence \eqref{eq:consistency} is not satisfied.
\end{proof}

\begin{remark}
In practice, for every irrational number $\generalReal$ and every $\seqLength$ there will be a rational number $\hat\generalReal$ that approximates $\generalReal$ so closely that for $\seqLength$ observations of $\noiselessSignal$ and $\noisySignal$, it does not matter if $\noiselessIntensity/\noiseIntensity$ is $\generalReal$ or $\hat\generalReal$. Nonetheless, if the floating point accuracy is increased as $\seqLength$ grows and $\noiselessIntensity/\noiseIntensity$ remains irrational and fixed (but evaluated more and more accurately), we can still expect that the limit statements of Theorem~\ref{theorem:iid-denoising}-\ref{item:iid-denoising-countable} hold. For numerical evaluations up to a fixed $\seqLength$ and with fixed floating point precision, however, it is important that we do not have $\noiselessIntensity/\noiseIntensity \approx \generalNatural_1/\generalNatural_2$ where $\generalNatural_1$ and $\generalNatural_2$ are too small numbers. That is, some irrational choices of $\noiselessIntensity/\noiseIntensity$ could be ``almost rational'' in the sense that there is some $\distributionIndex \neq \distributionIndex_0$ with non-negligible $\distributionIndexWeight(\distributionIndex_0)$ which satisfies $\jointDist_{\noiselessSignal,\noisySignal | \distributionIndex_0} \approx \jointDist_{\noiselessSignal,\noisySignal | \distributionIndex}$. As can be seen in the numerical results reported in this section, our choice of $\noiselessIntensity/\noiseIntensity = \sqrt{3/2}$ is not ``almost rational'' in this sense.
\end{remark}

In order to compute \eqref{eq:mismatched-denoising}, we need to be able to compute $\universalDist_{\noiselessSignal^\seqLength | \noisySignal^\seqLength = \noisySignalValue^\seqLength}$ for any given $\noisySignalValue^\seqLength$. Due to Lemma~\ref{lemma:simplified-additive}, in the case of additive losses, it is sufficient to compute $\universalDist_{\noiselessSignal_\seqIndex | \noisySignal^\seqLength = \noisySignalValue^\seqLength}$. In order to be able to evaluate the denoiser's regret, we also need to compute \eqref{eq:optimal-denoising}. Since $\noiselessSignal^\seqLength$ and $\noisySignal^\seqLength$ are i.i.d. for instances of Problem~\ref{problem:iid} (which Example~\ref{example:poisson} satisfies), we only need to be able to compute $\jointDist_{\noiselessSignal | \noisySignal, \distributionIndex_0}$ regardless of additivity of the loss.

\begin{lemma}
\label{lemma:poisson-likelihoods}
\emph{(Likelihoods in the Poisson example).}
For Example~\ref{example:poisson}, we have
\begin{align}
\label{eq:poisson-likelihoods-true}
\jointDist_{\noiselessSignal | \noisySignal, \numNoiseless, \numNoise}(\noiselessSignalValue|\noisySignalValue)
&=
\indicator{\noisySignalValue \geq \noiselessSignalValue}
\frac{\poissonpmf{\noiselessSignalValue}{\numNoiseless\noiselessIntensity}
      \poissonpmf{\noisySignalValue-\noiselessSignalValue}{\numNoise\noiseIntensity}}
     {\poissonpmf{\noisySignalValue}{\numNoiseless\noiselessIntensity+\numNoise\noiseIntensity}}
\\
\label{eq:poisson-likelihoods-mixture}
\universalDist_{\noiselessSignal^\seqLength | \noisySignal^\seqLength}(\noiselessSignalValue^\seqLength | \noisySignalValue^\seqLength)
&=
\indicator{\forall \seqIndex \in \{1, \dots, \seqLength\}:~ \noisySignalValue_\seqIndex \geq \noiselessSignalValue_\seqIndex}
\frac{
  \sum_{\numNoiseless=1}^\infty
  \sum_{\numNoise=1}^\infty
  \distributionIndexWeight(\numNoiseless,\numNoise)
  \prod_{\seqIndex=1}^\seqLength
    \poissonpmf{\noiselessSignalValue_\seqIndex}{\numNoiseless\noiselessIntensity}
    \poissonpmf{\noisySignalValue_\seqIndex-\noiselessSignalValue_\seqIndex}{\numNoise\noiseIntensity}
}{
  \sum_{\numNoiseless=1}^\infty
  \sum_{\numNoise=1}^\infty
  \distributionIndexWeight(\numNoiseless,\numNoise)
  \prod_{\seqIndex=1}^\seqLength
    \poissonpmf{\noisySignalValue_\seqIndex}{\numNoiseless\noiselessIntensity+\numNoise\noiseIntensity}
}
\\
\label{eq:poisson-likelihoods-mixture-marginal}
\universalDist_{\noiselessSignal_\seqIndex | \noisySignal^\seqLength}(\noiselessSignalValue_\seqIndex |  \noisySignalValue^\seqLength)
&=
\indicator{\noisySignalValue_\seqIndex \geq \noiselessSignalValue_\seqIndex}
\frac{
  \sum_{\numNoiseless=1}^\infty
  \sum_{\numNoise=1}^\infty
    \distributionIndexWeight(\numNoiseless,\numNoise)
    \poissonpmf{\noiselessSignalValue_\seqIndex}{\numNoiseless\noiselessIntensity}
    \poissonpmf{\noisySignalValue_\seqIndex-\noiselessSignalValue_\seqIndex}{\numNoise\noiseIntensity}
    \prod_{\substack{\seqIndex'=1\\ \seqIndex' \neq \seqIndex}}^\seqLength
      \poissonpmf{\noisySignalValue_{\seqIndex'}}{\numNoiseless\noiselessIntensity+\numNoise\noiseIntensity}
}{
  \sum_{\numNoiseless=1}^\infty
  \sum_{\numNoise=1}^\infty
  \distributionIndexWeight(\numNoiseless,\numNoise)
  \prod_{\seqIndex=1}^\seqLength
    \poissonpmf{\noisySignalValue_\seqIndex}{\numNoiseless\noiselessIntensity+\numNoise\noiseIntensity}
},
\end{align}
where $\poissonpmf{\generalNatural}{\generalIntensity}$ denotes the Poisson p.m.f., i.e., the probability that a Poisson distributed random variable with mean $\generalIntensity$ equals $\generalNatural$.
\end{lemma}
\begin{proof}
From the description of Example~\ref{example:poisson}, it is clear that
\begin{align*}
\jointDist_{\noiselessSignal|\numNoiseless, \numNoise}(\noiselessSignalValue)
&=
\poissonpmf{\noiselessSignalValue}{\numNoiseless\noiselessIntensity}
\\
\jointDist_{\additiveNoise|\numNoiseless, \numNoise}(\additiveNoiseValue)
&=
\poissonpmf{\additiveNoiseValue}{\numNoise\noiseIntensity}.
\end{align*}
We can infer from this
\begin{align}
\nonumber
\jointDist_{\noisySignal|\noiselessSignal, \numNoiseless, \numNoise}(\noisySignalValue|\noiselessSignalValue)
&=
\indicator{\noisySignalValue \geq \noiselessSignalValue}
\poissonpmf{\noisySignalValue-\noiselessSignalValue}{\numNoise\noiseIntensity}
\\
\label{eq:poisson-likelihoods-true-joint}
\jointDist_{\noiselessSignal, \noisySignal | \numNoiseless, \numNoise}(\noiselessSignalValue, \noisySignalValue)
&=
\jointDist_{\noiselessSignal|\numNoiseless, \numNoise}(\noiselessSignalValue)
\jointDist_{\noisySignal|\noiselessSignal, \numNoiseless, \numNoise}(\noisySignalValue|\noiselessSignalValue)
=
\indicator{\noisySignalValue \geq \noiselessSignalValue}
\poissonpmf{\noiselessSignalValue}{\numNoiseless\noiselessIntensity}
\poissonpmf{\noisySignalValue-\noiselessSignalValue}{\numNoise\noiseIntensity}.
\end{align}
Due to the known additivity property of Poisson random variables, we also have
\begin{equation}
\label{eq:poisson-likelihoods-true-noisy}
\jointDist_{\noisySignal|\numNoiseless, \numNoise}(\noisySignalValue)
=
\poissonpmf{\noisySignalValue}{\numNoiseless\noiselessIntensity+\numNoise\noiseIntensity}.
\end{equation}
\eqref{eq:poisson-likelihoods-true} then follows directly from \eqref{eq:poisson-likelihoods-true-joint} and \eqref{eq:poisson-likelihoods-true-noisy}. For the mixture distribution, we use \eqref{eq:poisson-likelihoods-true-joint} to compute
\begin{align}
\nonumber
\universalDist_{\noiselessSignal^\seqLength,\noisySignal^\seqLength}(\noiselessSignalValue^\seqLength, \noisySignalValue^\seqLength)
&=
\sum_{\numNoiseless=1}^\infty
\sum_{\numNoise=1}^\infty
\distributionIndexWeight(\numNoiseless,\numNoise)
\prod_{\seqIndex=1}^\seqLength
  \jointDist_{\noiselessSignal, \noisySignal | \numNoiseless, \numNoise}(\noiselessSignalValue_\seqIndex, \noisySignalValue_\seqIndex)
\\
\nonumber
&=
\sum_{\numNoiseless=1}^\infty
\sum_{\numNoise=1}^\infty
\distributionIndexWeight(\numNoiseless,\numNoise)
\prod_{\seqIndex=1}^\seqLength
  \poissonpmf{\noiselessSignalValue_\seqIndex}{\numNoiseless\noiselessIntensity}
  \poissonpmf{\noisySignalValue_\seqIndex-\noiselessSignalValue_\seqIndex}{\numNoise\noiseIntensity}
  \indicator{\noisySignalValue_\seqIndex \geq \noiselessSignalValue_\seqIndex}
\\
\label{eq:poisson-likelihoods-mixture-joint}
&=
\indicator{\forall \seqIndex \in \{1, \dots, \seqLength\}:~ \noisySignalValue_\seqIndex \geq \noiselessSignalValue_\seqIndex}
\sum_{\numNoiseless=1}^\infty
\sum_{\numNoise=1}^\infty
\distributionIndexWeight(\numNoiseless,\numNoise)
\prod_{\seqIndex=1}^\seqLength
  \poissonpmf{\noiselessSignalValue_\seqIndex}{\numNoiseless\noiselessIntensity}
  \poissonpmf{\noisySignalValue_\seqIndex-\noiselessSignalValue_\seqIndex}{\numNoise\noiseIntensity},
\end{align}
and similarly we use \eqref{eq:poisson-likelihoods-true-noisy} to compute
\begin{align}
\nonumber
\universalDist_{\noisySignal^\seqLength}(\noisySignalValue^\seqLength)
&=
\sum_{\numNoiseless=1}^\infty
\sum_{\numNoise=1}^\infty
\distributionIndexWeight(\numNoiseless,\numNoise)
\prod_{\seqIndex=1}^\seqLength
  \jointDist_{\noisySignal | \numNoiseless, \numNoise}(\noisySignalValue_\seqIndex)
\\
\label{eq:poisson-likelihoods-mixture-noisy}
&=
\sum_{\numNoiseless=1}^\infty
\sum_{\numNoise=1}^\infty
\distributionIndexWeight(\numNoiseless,\numNoise)
\prod_{\seqIndex=1}^\seqLength
  \poissonpmf{\noisySignalValue_\seqIndex}{\numNoiseless\noiselessIntensity+\numNoise\noiseIntensity}.
\end{align}
\eqref{eq:poisson-likelihoods-mixture} then follows directly from \eqref{eq:poisson-likelihoods-mixture-joint} and \eqref{eq:poisson-likelihoods-mixture-noisy}. Finally, we use \eqref{eq:poisson-likelihoods-true} and \eqref{eq:poisson-likelihoods-true-noisy} to obtain
\begin{align*}
\universalDist_{\noiselessSignal_\seqIndex,\noisySignal^\seqLength}(\noiselessSignalValue_\seqIndex, \noisySignalValue^\seqLength)
&=
\sum_{\noiselessSignalValue_1 \in \noiselessSignalAlphabet}
\cdots
\sum_{\noiselessSignalValue_{\seqIndex-1} \in \noiselessSignalAlphabet}
\sum_{\noiselessSignalValue_{\seqIndex+1} \in \noiselessSignalAlphabet}
\cdots
\sum_{\noiselessSignalValue_\seqLength \in \noiselessSignalAlphabet}
  \universalDist_{\noiselessSignal^\seqLength,\noisySignal^\seqLength}(\noiselessSignalValue^\seqLength, \noisySignalValue^\seqLength)
\\
&=
\sum_{\noiselessSignalValue_1 \in \noiselessSignalAlphabet}
\cdots
\sum_{\noiselessSignalValue_{\seqIndex-1} \in \noiselessSignalAlphabet}
\sum_{\noiselessSignalValue_{\seqIndex+1} \in \noiselessSignalAlphabet}
\cdots
\sum_{\noiselessSignalValue_\seqLength \in \noiselessSignalAlphabet}
\sum_{\numNoiseless=1}^\infty
\sum_{\numNoise=1}^\infty
  \distributionIndexWeight(\numNoiseless,\numNoise)
  \jointDist_{\noiselessSignal, \noisySignal | \numNoiseless, \numNoise}^\seqLength(\noiselessSignalValue^\seqLength, \noisySignalValue^\seqLength)
\\
&=
\sum_{\numNoiseless=1}^\infty
\sum_{\numNoise=1}^\infty
  \distributionIndexWeight(\numNoiseless,\numNoise)
  \sum_{\noiselessSignalValue_1 \in \noiselessSignalAlphabet}
  \cdots
  \sum_{\noiselessSignalValue_{\seqIndex-1} \in \noiselessSignalAlphabet}
  \sum_{\noiselessSignalValue_{\seqIndex+1} \in \noiselessSignalAlphabet}
  \cdots
  \sum_{\noiselessSignalValue_\seqLength \in \noiselessSignalAlphabet}
  \jointDist_{\noiselessSignal, \noisySignal | \numNoiseless, \numNoise}^\seqLength(\noiselessSignalValue^\seqLength, \noisySignalValue^\seqLength)
\\
&=
\sum_{\numNoiseless=1}^\infty
\sum_{\numNoise=1}^\infty
  \distributionIndexWeight(\numNoiseless,\numNoise)
  \jointDist_{\noiselessSignal_\seqIndex, \noisySignal^\seqLength | \numNoiseless, \numNoise}(\noiselessSignalValue_\seqIndex, \noisySignalValue^\seqLength)
\\
&=
\sum_{\numNoiseless=1}^\infty
\sum_{\numNoise=1}^\infty
  \distributionIndexWeight(\numNoiseless,\numNoise)
  \jointDist_{\noisySignal^\seqLength | \numNoiseless, \numNoise}^\seqLength(\noisySignalValue^\seqLength)
  \jointDist_{\noiselessSignal_\seqIndex | \noisySignal_\seqIndex, \numNoiseless, \numNoise}(\noiselessSignalValue_\seqIndex, \noisySignalValue_\seqIndex)
\\
&=
\indicator{\noisySignalValue_\seqIndex \geq \noiselessSignalValue_\seqIndex}
\sum_{\numNoiseless=1}^\infty
\sum_{\numNoise=1}^\infty
  \distributionIndexWeight(\numNoiseless,\numNoise)
  \left(
    \prod_{\seqIndex'=1}^\seqLength
      \poissonpmf{\noisySignalValue_{\seqIndex'}}{\numNoiseless\noiselessIntensity+\numNoise\noiseIntensity}
  \right)
  \frac{\poissonpmf{\noiselessSignalValue_\seqIndex}{\numNoiseless\noiselessIntensity}
        \poissonpmf{\noisySignalValue_\seqIndex-\noiselessSignalValue_\seqIndex}{\numNoise\noiseIntensity}}
       {\poissonpmf{\noisySignalValue_\seqIndex}{\numNoiseless\noiselessIntensity+\numNoise\noiseIntensity}}
\\
&=
\indicator{\noisySignalValue_\seqIndex \geq \noiselessSignalValue_\seqIndex}
\sum_{\numNoiseless=1}^\infty
\sum_{\numNoise=1}^\infty
  \distributionIndexWeight(\numNoiseless,\numNoise)
  \poissonpmf{\noiselessSignalValue_\seqIndex}{\numNoiseless\noiselessIntensity}
  \poissonpmf{\noisySignalValue_\seqIndex-\noiselessSignalValue_\seqIndex}{\numNoise\noiseIntensity}
  \prod_{\substack{\seqIndex'=1\\ \seqIndex' \neq \seqIndex}}^\seqLength
    \poissonpmf{\noisySignalValue_{\seqIndex'}}{\numNoiseless\noiselessIntensity+\numNoise\noiseIntensity},
\end{align*}
which together with \eqref{eq:poisson-likelihoods-mixture-noisy} yields \eqref{eq:poisson-likelihoods-mixture-marginal}.
\end{proof}

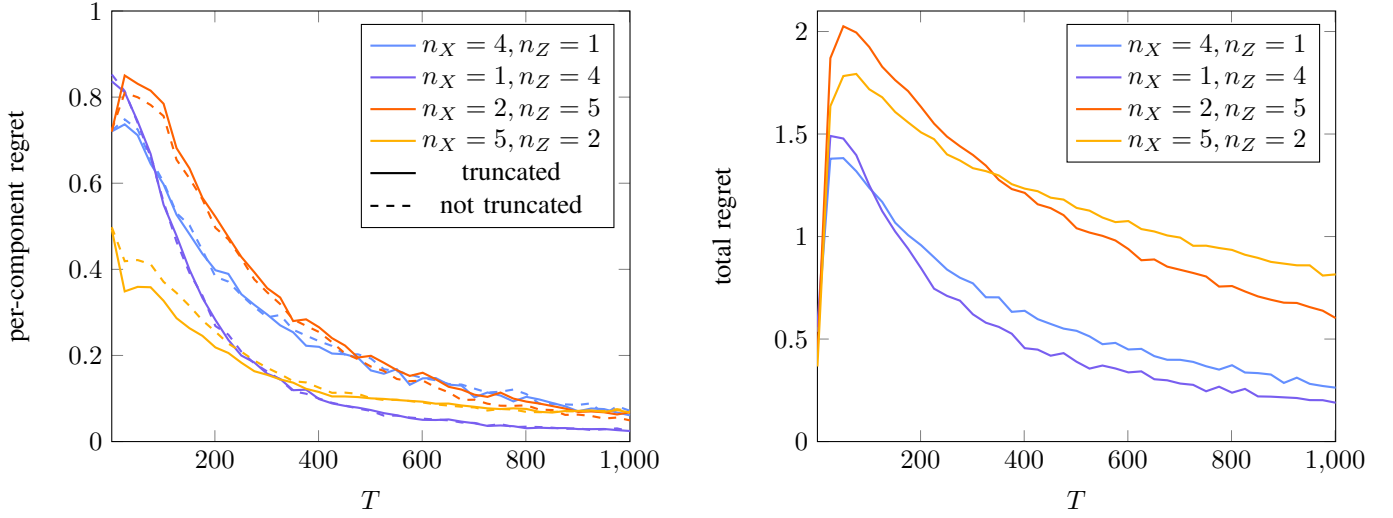
\begin{figure}
\begin{subfigure}[t]{.45\textwidth}
\begin{tikzpicture}
\begin{axis}[
  ymin=0,
  ymax=1,
  xmin=1,
  xmax=1001,
  xlabel={$\seqLength$},
  ylabel={per-component regret},
  legend pos=north east,
]
  \addlegendimage{color=plot1, thick};
  \addlegendentry{$\numNoiseless=4, \numNoise=1$};
  \addlegendimage{color=plot2, thick};
  \addlegendentry{$\numNoiseless=1, \numNoise=4$};
  \addlegendimage{color=plot4, thick};
  \addlegendentry{$\numNoiseless=2, \numNoise=5$};
  \addlegendimage{color=plot5, thick};
  \addlegendentry{$\numNoiseless=5, \numNoise=2$};
  \addlegendimage{color=black, thick};
  \addlegendentry{truncated};
  \addlegendimage{color=black, thick, dashed};
  \addlegendentry{not truncated};
  \addplot[color=plot1, thick] table[x=seq_len, y=avg_regret_per_comp, col sep=comma] {poisson_results_distancetrunc5_4_1.csv};
  \addplot[color=plot1, thick, dashed] table[x=seq_len, y=avg_regret_per_comp, col sep=comma] {poisson_results_distance_4_1.csv};
  \addplot[color=plot2, thick] table[x=seq_len, y=avg_regret_per_comp, col sep=comma] {poisson_results_distancetrunc5_1_4.csv};
  \addplot[color=plot2, thick, dashed] table[x=seq_len, y=avg_regret_per_comp, col sep=comma] {poisson_results_distance_1_4.csv};
  \addplot[color=plot4, thick] table[x=seq_len, y=avg_regret_per_comp, col sep=comma] {poisson_results_distancetrunc5_2_5.csv};
  \addplot[color=plot4, thick, dashed] table[x=seq_len, y=avg_regret_per_comp, col sep=comma] {poisson_results_distance_2_5.csv};
  \addplot[color=plot5, thick] table[x=seq_len, y=avg_regret_per_comp, col sep=comma] {poisson_results_distancetrunc5_5_2.csv};
  \addplot[color=plot5, thick, dashed] table[x=seq_len, y=avg_regret_per_comp, col sep=comma] {poisson_results_distance_5_2.csv};
\end{axis}
\end{tikzpicture}
\caption{Per-component regret for truncated and untruncated distance loss.}
\end{subfigure}
\hspace{.05\textwidth}
\begin{subfigure}[t]{.45\textwidth}
\begin{tikzpicture}
\begin{axis}[
  ymin=0,
  ymax=2.1,
  xmin=1,
  xmax=1001,
  xlabel={$\seqLength$},
  ylabel={total regret},
  legend pos=north east,
]
  \addplot[color=plot1, thick] table[x=seq_len, y=avg_regret, col sep=comma] {poisson_results_selfinformation_4_1.csv};
  \addplot[color=plot2, thick] table[x=seq_len, y=avg_regret, col sep=comma] {poisson_results_selfinformation_1_4.csv};
  \addplot[color=plot4, thick] table[x=seq_len, y=avg_regret, col sep=comma] {poisson_results_selfinformation_2_5.csv};
  \addplot[color=plot5, thick] table[x=seq_len, y=avg_regret, col sep=comma] {poisson_results_selfinformation_5_2.csv};
  \legend{{$\numNoiseless=4, \numNoise=1$},
          {$\numNoiseless=1, \numNoise=4$},
          {$\numNoiseless=2, \numNoise=5$},
          {$\numNoiseless=5, \numNoise=2$},};
\end{axis}
\end{tikzpicture}
\caption{Total regret for self-information loss.}
\end{subfigure}
\caption{Regrets incurred by the denoiser in Example~\ref{example:poisson} for various values of $\numNoiseless$ and $\numNoise$, each data point averaged over 3,000 runs.}
\label{fig:poisson}
\end{figure}

We have implemented the denoising strategy given by \eqref{eq:mismatched-denoising} and \eqref{eq:universal-dist} with
\[
  \distributionIndexWeight(\numNoiseless, \numNoise)
  :=
  2^{-\numNoiseless-\numNoise},
\]
as well as the optimal denoiser \eqref{eq:optimal-denoising}. For the intensities, we use $\noiselessIntensity = 1$ and $\noiseIntensity = \sqrt{2/3}$. We have used \eqref{eq:poisson-likelihoods-mixture-joint} along with Lemma~\ref{lemma:optimal-self-information} to compute the denoising estimate for the case of self-information loss and \eqref{eq:poisson-likelihoods-true} along with Lemma~\ref{lemma:optimal-self-information} to compute the corresponding optimal denoising estimate. It should be noted that while computing $\estSignal(\noiselessSignalValue^\seqIndex)$ is feasible for any given value of $\noiselessSignalValue^\seqIndex$, due to the number of possibilities for $\noiselessSignalValue^\seqIndex$, it is not tractable to compute the entirety of $\estSignal$ even in the case of moderate $\seqLength$. However, the loss for a concrete realization of $\noiselessSignal^\seqIndex$ can be computed in a simulation. For the case of additive losses, we use \eqref{eq:poisson-likelihoods-mixture-joint} along with Lemma~\ref{lemma:simplified-additive} to compute the denoising estimate (note that it is guaranteed in Example~\ref{example:poisson} that we have $\noiselessSignal_\seqIndex \leq \noisySignal_\seqIndex$ so it is tractable to iterate over all the possible choices in \eqref{eq:mismatched-denoising}) for the \emph{distance loss}
\[
  \lossFunc_\mathrm{dist}(\estSignalValue, \noiselessSignalValue)
  :=
  \absoluteValue{\estSignalValue - \noiselessSignalValue}
\]
as well as for the \emph{truncated distance loss}
\[
  \lossFunc_{\mathrm{dist}, 5}(\estSignalValue, \noiselessSignalValue)
  :=
  \min\Big(\lossFunc_\mathrm{dist}(\estSignalValue, \noiselessSignalValue), 5\Big).
\]
Note that we have a theoretical guarantee on the regret only for the case of the truncated distance loss. Similarly, we use \eqref{eq:poisson-likelihoods-true} and Lemma~\ref{lemma:simplified-additive} to compute the optimal denoising estimate and the regret. The documented Python code used to generate the plots shown in this section is available as an electronic supplement with this paper. In Fig.~\ref{fig:poisson}, we show the per-component regret for the distance loss and the total regret for the self-information loss case. As Theorem~\ref{theorem:iid-denoising} predicts, in the case of truncated distance loss, the per-component regret becomes small as $\seqLength$ increases and in the case of self-information loss, the total regret becomes small as $\seqLength$ increases. We can additionally see that the (untruncated) distance loss effectively behaves quite similar to its truncated version although we do not have a theoretical guarantee for this case. As can be expected from the $\distributionIndexWeight$ we use, there is a rough trend that higher values of $\numNoiseless$ and $\numNoise$ result in slower convergence of the regret.

\section{Acknowledgments}
This research was supported in part by the Australian Research Council under Project DP230101493, and by The University of Melbourne’s Research Computing Services and the Petascale Campus. The authors would also like to thank Samuel H. Wilks for many useful discussions on the topic.

\Gls*{AI} has been used to generate code related to the command line user interface of the numerical simulations discussed in Section~\ref{sec:poisson} and to brainstorm ideas for which probability distributions could be used in Example~\ref{example:poisson}. No technical results and no text that appears in the paper have been generated with the involvement of \gls*{AI}.

\appendix
\subsection{Proof of Theorem~\ref{theorem:oneshot}}
For the sake of completeness, we first prove the following elementary lemma which we will use in our proof of the theorem.
\label{appendix:oneshot-proof}
\begin{lemma}
\label{lemma:conditional-kl}
Let $\generalRVOne$ and $\generalRVTwo$ be two random variables, let $\generalPmeasureOne_\generalRVTwo$ and $\generalPmeasureTwo_\generalRVTwo$ be two probability measures for $\generalRVTwo$, let $\generalPmeasureOne_{\generalRVOne|\generalRVTwo}$ and $\generalPmeasureTwo_{\generalRVOne|\generalRVTwo}$ be two stochastic kernels describing transitions from $\generalRVTwo$ to $\generalRVOne$, let $\generalPmeasureOne_{\generalRVOne,\generalRVTwo}$ be the joint probability measure for $\generalRVOne$ and $\generalRVTwo$ induced by $\generalPmeasureOne_\generalRVTwo$ and $\generalPmeasureOne_{\generalRVOne|\generalRVTwo}$, and let and $\generalPmeasureTwo_{\generalRVOne,\generalRVTwo}$ be the joint probability measure induced by $\generalPmeasureTwo_\generalRVTwo$ and $\generalPmeasureTwo_{\generalRVOne|\generalRVTwo}$. Then
\[
\Expectation_{\generalPmeasureOne_\generalRVTwo}
  \kldiv{\generalPmeasureOne_{\generalRVOne|\generalRVTwo}(\cdot|\generalRVTwo)}{\generalPmeasureTwo_{\generalRVOne|\generalRVTwo}(\cdot|\generalRVTwo)}
=
\kldiv{\generalPmeasureOne_{\generalRVOne,\generalRVTwo}}{\generalPmeasureTwo_{\generalRVOne,\generalRVTwo}}
-
\kldiv{\generalPmeasureOne_{\generalRVTwo}}{\generalPmeasureTwo_{\generalRVTwo}}.
\]
\end{lemma}
\begin{proof}
We have
\begin{align*}
\Expectation_{\generalPmeasureOne_\generalRVTwo}
  \kldiv{\generalPmeasureOne_{\generalRVOne|\generalRVTwo}(\cdot|\generalRVTwo)}{\generalPmeasureTwo_{\generalRVOne|\generalRVTwo}(\cdot|\generalRVTwo)}
\overset{(a)}&{=}
\Expectation_{\generalPmeasureOne_{\generalRVTwo}}
\Expectation_{\generalPmeasureOne_{\generalRVOne|\generalRVTwo}}
  \log\rnderiv{\generalPmeasureOne_{\generalRVOne|\generalRVTwo}(\cdot|\generalRVTwo)}
              {\generalPmeasureTwo_{\generalRVOne|\generalRVTwo}(\cdot|\generalRVTwo)}
      (\generalRVOne)
\\
\overset{(b)}&{=}
\Expectation_{\generalPmeasureOne_{\generalRVOne,\generalRVTwo}}
  \log\left(
    \rnderiv{\generalPmeasureOne_{\generalRVOne|\generalRVTwo}(\cdot|\generalRVTwo)}
            {\generalPmeasureTwo_{\generalRVOne|\generalRVTwo}(\cdot|\generalRVTwo)}
    (\generalRVOne)
    \rnderiv{\generalPmeasureOne_{\generalRVTwo}}
            {\generalPmeasureTwo_{\generalRVTwo}}
    (\generalRVTwo)
  \right)
-
\Expectation_{\generalPmeasureOne_{\generalRVTwo}}
  \log\rnderiv{\generalPmeasureOne_{\generalRVTwo}}
              {\generalPmeasureTwo_{\generalRVTwo}}
      (\generalRVTwo)
\\
\overset{(c)}&{=}
\Expectation_{\generalPmeasureOne_{\generalRVOne,\generalRVTwo}}
  \log
    \rnderiv{\generalPmeasureOne_{\generalRVOne,\generalRVTwo}}
            {\generalPmeasureTwo_{\generalRVOne,\generalRVTwo}}
    (\generalRVOne,\generalRVTwo)
-
\Expectation_{\generalPmeasureOne_{\generalRVTwo}}
  \log\rnderiv{\generalPmeasureOne_{\generalRVTwo}}
              {\generalPmeasureTwo_{\generalRVTwo}}
      (\generalRVTwo)
\\
\overset{(a)}&{=}
\kldiv{\generalPmeasureOne_{\generalRVOne,\generalRVTwo}}{\generalPmeasureTwo_{\generalRVOne,\generalRVTwo}}
-
\kldiv{\generalPmeasureOne_{\generalRVTwo}}{\generalPmeasureTwo_{\generalRVTwo}},
\end{align*}
where in both steps labeled (a) we have used the definition of Kullback-Leibler divergence, in step (b) we have added zero, and in step (c) we have used~\cite[Section 6.3, Proposition 1]{rao2004measure}.
\end{proof}
\begin{proof}[Proof of Theorem~\ref{theorem:oneshot}]
To show \eqref{eq:oneshot-bounded-loss}, we first bound
\begin{align}
\nonumber
&\hphantom{{}={}}
\Expectation_{\jointDist_{\noiselessSignal|\noisySignal}}\left(
  \totalLoss(\estSignal,\noiselessSignal)
  -
  \totalLoss(\estSignalOpt,\noiselessSignal)
\right)
\\
\nonumber
&=
\Expectation_{\jointDist_{\noiselessSignal|\noisySignal}}\left(
  \totalLoss(\estSignal,\noiselessSignal)
  -
  \totalLoss(\estSignalOpt,\noiselessSignal)
\right)^+
-
\Expectation_{\jointDist_{\noiselessSignal|\noisySignal}}\left(
  \totalLoss(\estSignal,\noiselessSignal)
  -
  \totalLoss(\estSignalOpt,\noiselessSignal)
\right)^-
\\
\nonumber
\overset{(a)}&{=}
\int_0^{\totalLoss_{\max}}
  \jointDist_{\noiselessSignal|\noisySignal}\left(
    \totalLoss(\estSignal,\noiselessSignal)
    -
    \totalLoss(\estSignalOpt,\noiselessSignal)
    \geq
    \generalReal
  \right)
d\generalReal
-
\int_0^{\totalLoss_{\max}}
  \jointDist_{\noiselessSignal|\noisySignal}\left(
    \totalLoss(\estSignal,\noiselessSignal)
    -
    \totalLoss(\estSignalOpt,\noiselessSignal)
    \leq
    -\generalReal
  \right)
d\generalReal
\\
\nonumber
\overset{(b)}&{\leq}
\begin{multlined}[t]
\int_0^{\totalLoss_{\max}}\left(
  \universalDist_{\noiselessSignal|\noisySignal}\left(
    \totalLoss(\estSignal,\noiselessSignal)
    -
    \totalLoss(\estSignalOpt,\noiselessSignal)
    \geq
    \generalReal
  \right)
  +
  \tvdist{\jointDist-\universalDist}
\right)d\generalReal
\\-
\int_0^{\totalLoss_{\max}}\left(
  \universalDist_{\noiselessSignal|\noisySignal}\left(
    \totalLoss(\estSignal,\noiselessSignal)
    -
    \totalLoss(\estSignalOpt,\noiselessSignal)
    \leq
    -\generalReal
  \right)
  -
  \tvdist{\jointDist-\universalDist}
\right)d\generalReal
\end{multlined}
\\
\nonumber
&=
\Expectation_{\universalDist_{\noiselessSignal|\noisySignal}}\left(
  \totalLoss(\estSignal,\noiselessSignal)
  -
  \totalLoss(\estSignalOpt,\noiselessSignal)
\right)^+
-
\Expectation_{\universalDist_{\noiselessSignal|\noisySignal}}\left(
  \totalLoss(\estSignal,\noiselessSignal)
  -
  \totalLoss(\estSignalOpt,\noiselessSignal)
\right)^-
+
2\tvdist{\jointDist-\universalDist}
\\
\label{eq:expectation-tvdist}
&=
\Expectation_{\universalDist_{\noiselessSignal|\noisySignal}}\left(
  \totalLoss(\estSignal,\noiselessSignal)
  -
  \totalLoss(\estSignalOpt,\noiselessSignal)
\right)
+
2\tvdist{\jointDist-\universalDist},
\end{align}
where (a) is due to the assumption that $\totalLoss$ only takes values in $[0,\totalLoss_{\max}]$ and (b) holds due to the definition of variational distance. Using this, we can bound

\begin{align*}
\regret
&=
\Expectation_{\jointDist_{\noisySignal}}
\Expectation_{\jointDist_{\noiselessSignal|\noisySignal}}\left(
  \totalLoss(\estSignal,\noiselessSignal)
  -
  \totalLoss(\estSignalOpt,\noiselessSignal)
\right)
\\
\overset{\eqref{eq:expectation-tvdist}}&{\leq}
\Expectation_{\jointDist_{\noisySignal}}\left(
  \Expectation_{\universalDist_{\noiselessSignal|\noisySignal}}\left(
    \totalLoss(\estSignal,\noiselessSignal)
    -
    \totalLoss(\estSignalOpt,\noiselessSignal)
  \right)
  +
  2 \totalLoss_{\max}
  \tvdist{
    \jointDist_{\noiselessSignal|\noisySignal}
    -
    \universalDist_{\noiselessSignal|\noisySignal}
  }
\right)
\\
\overset{\eqref{eq:mismatched-denoising}}&{\leq}
2 \totalLoss_{\max}
\Expectation_{\jointDist_{\noisySignal}}
  \tvdist{
    \jointDist_{\noiselessSignal|\noisySignal}
    -
    \universalDist_{\noiselessSignal|\noisySignal}
  }
\\
\overset{(a)}&{\leq}
2 \totalLoss_{\max}
\Expectation_{\jointDist_{\noisySignal}}\sqrt{
  \frac{1}{2}
  \kldiv{
    \jointDist_{\noiselessSignal|\noisySignal}
  }{
    \universalDist_{\noiselessSignal|\noisySignal}
  }
}
\\
\overset{(b)}&{\leq}
2 \totalLoss_{\max}
\sqrt{
  \frac{1}{2}
  \Expectation_{\jointDist_{\noisySignal}}
  \kldiv{
    \jointDist_{\noiselessSignal|\noisySignal}
  }{
    \universalDist_{\noiselessSignal|\noisySignal}
  }
}
\\
\overset{(c)}&{=}
2 \totalLoss_{\max}
\sqrt{
  \frac{1}{2}\left(
    \kldiv{\jointDist_{\noiselessSignal, \noisySignal}}{\universalDist_{\noiselessSignal, \noisySignal}}
    -
    \kldiv{\jointDist_{\noisySignal}}{\universalDist_{\noisySignal}}
  \right)
}
\end{align*}
where step (a) follows by Pinsker's inequality, in step (b) we have used Jensen's inequality and the concavity of the square root function, and step (c) follows by Lemma~\ref{lemma:conditional-kl}. This concludes the proof of \eqref{eq:oneshot-bounded-loss}.

To show \eqref{eq:oneshot-log-loss}, we first need to determine a $\estSignalOpt$ which satisfies \eqref{eq:optimal-denoising} and a $\estSignal$ that satisfies \eqref{eq:mismatched-denoising}. To this end, we prove the following lemma.

\begin{lemma}
\label{lemma:optimal-self-information}
\emph{(Optimal estimate for the self-information loss).}
If $\totalLoss$ is the self-information loss, then $\estSignalValue = \universalDist_{\noiselessSignal | \noisySignal = \noisySignalValue}$ is the unique minimum of
\[
  \Expectation_{\universalDist_{\noiselessSignal | \noisySignal = \noisySignalValue}} \totalLoss(\estSignalValue, \noiselessSignal)
  =
  -\Expectation_{\universalDist_{\noiselessSignal | \noisySignal = \noisySignalValue}} \log \estSignalValue(\noiselessSignal).
\]
\end{lemma}
\begin{proof}
The proof proceeds similarly as in the online prediction case~\cite{merhav1998universal}. We have
\begin{align*}
-\Expectation_{\universalDist_{\noiselessSignal | \noisySignal = \noisySignalValue}} \log \estSignalValue(\noiselessSignal)
\overset{(a)}&{=}
-
\Expectation_{\universalDist_{\noiselessSignal | \noisySignal = \noisySignalValue}} \log \universalDist_{\noiselessSignal | \noisySignal = \noisySignalValue}(\noiselessSignal)
+
\Expectation_{\universalDist_{\noiselessSignal | \noisySignal = \noisySignalValue}} \log \frac{\universalDist_{\noiselessSignal | \noisySignal = \noisySignalValue}(\noiselessSignal)}{\estSignalValue(\noiselessSignal)}
\\
&=
\entropy{\universalDist_{\noiselessSignal | \noisySignal = \noisySignalValue}}
+
\kldiv{\universalDist_{\noiselessSignal | \noisySignal = \noisySignalValue}}{\estSignalValue},
\end{align*}
where step (a) is an addition of zero and $\entropy{\cdot}$ is Shannon entropy (note that in matching the notation of Kullback-Leibler divergence but different from common notation we indicate the probability distribution instead of the random variable in the argument). The entropy term is independent of $\estSignalValue$ and the Kullback-Leibler divergence takes its minimum value $0$ if and only if $\estSignalValue = \universalDist_{\noiselessSignal | \noisySignal = \noisySignalValue}$.
\end{proof}

So clearly, we have that $\estSignal = \universalDist_{\noiselessSignal|\noisySignal}(\cdot|\noisySignal)$ satisfies \eqref{eq:mismatched-denoising}. Since $\jointDist$ can be substituted for $\universalDist$ in Lemma~\ref{lemma:optimal-self-information}, we can also conclude that $\estSignalOpt = \jointDist_{\noiselessSignal | \noisySignal}(\cdot|\noisySignal)$ satisfies \eqref{eq:optimal-denoising}.

Substituting both of these into definition~\eqref{eq:def-regret}, we obtain
\begin{align*}
\regret
&=
\Expectation_{\jointDist_{\noiselessSignal,\noisySignal}}\left(
  - \log \universalDist_{\noiselessSignal | \noisySignal}(\noiselessSignal|\noisySignal)
  + \log \jointDist_{\noiselessSignal | \noisySignal}(\noiselessSignal|\noisySignal)
\right)
\\
&=
\Expectation_{\jointDist_{\noisySignal}}
\Expectation_{\jointDist_{\noiselessSignal|\noisySignal}}
  \log\frac{\jointDist_{\noiselessSignal | \noisySignal}(\noiselessSignal|\noisySignal)}
           {\universalDist_{\noiselessSignal | \noisySignal}(\noiselessSignal|\noisySignal)}
\\
\overset{(a)}&{=}
\Expectation_{\jointDist_{\noisySignal}}
  \kldiv{\jointDist_{\noiselessSignal | \noisySignal}}{\universalDist_{\noiselessSignal | \noisySignal}}
\\
\overset{(b)}&{=}
\kldiv{\jointDist_{\noiselessSignal, \noisySignal}}{\universalDist_{\noiselessSignal, \noisySignal}}
-
\kldiv{\jointDist_{\noisySignal}}{\universalDist_{\noisySignal}},
\end{align*}
where (a) follows by the definition of Kullback-Leibler divergence and (b) by Lemma~\ref{lemma:conditional-kl}. This concludes the proof of \eqref{eq:oneshot-log-loss}.
\end{proof}

\subsection{Proof of Lemma~\ref{lemma:kldiv-iid}}
\label{appendix:iid-proof}
First, we argue that condition \eqref{eq:consistency} allows us without loss of generality to assume that the slightly stronger condition
\begin{equation}
\label{eq:consistency-stronger}
  \forall \distributionIndex_1 \neq \distributionIndex_2 \in \distributionIndexDomain
  :~
  \jointDist_{\noisySignal | \distributionIndex_1}
  \neq
  \jointDist_{\noisySignal | \distributionIndex_2}
\end{equation}
holds. If \eqref{eq:consistency} holds but \eqref{eq:consistency-stronger} does not, we define an equivalence relation $\approx$ via
\[
  \distributionIndex_1
  \approx
  \distributionIndex_2
  :\Leftrightarrow
  \jointDist_{\noisySignal | \distributionIndex_1} = \jointDist_{\noisySignal | \distributionIndex_2}
\]
and let
\[
\distributionIndexDomain'
:=
\distributionIndexDomain/\approx
=
\left\{
  \left\{ \distributionIndex_2:~ \jointDist_{\noisySignal | \distributionIndex_1} = \jointDist_{\noisySignal | \distributionIndex_2} \right\}:~
  \distributionIndex_1 \in \distributionIndexDomain
\right\}
\]
be the quotient of $\distributionIndexDomain$ with respect to $\approx$. For $\distributionIndex'$, let
\[
\jointDist_{\noiselessSignal, \noisySignal | \distributionIndex'}
:=
\jointDist_{\noiselessSignal, \noisySignal | \distributionIndex},
\]
where $\distributionIndex$ is an element of the equivalence class $\distributionIndex'$ (it does not matter which element we choose from each class due to condition~\eqref{eq:consistency}). We define a p.m.f. $\distributionIndexWeight'$ on $\distributionIndexDomain'$ as
\[
\distributionIndexWeight'(\distributionIndex')
:=
\sum_{\distributionIndex \in \distributionIndex'}
  \distributionIndexWeight(\distributionIndex).
\]
Clearly, $\distributionIndexDomain'$ satisfies condition~\eqref{eq:consistency-stronger} and
\[
\sum_{\distributionIndex \in \distributionIndexDomain}
  \jointDist_{\noiselessSignal, \noisySignal | \distributionIndex}^\seqLength
  \distributionIndexWeight(\distributionIndex)
=
\sum_{\distributionIndex' \in \distributionIndexDomain'}
\sum_{\distributionIndex \in \distributionIndex'}
  \jointDist_{\noiselessSignal, \noisySignal | \distributionIndex}^\seqLength
  \distributionIndexWeight(\distributionIndex)
=
\sum_{\distributionIndex' \in \distributionIndexDomain}
  \jointDist_{\noiselessSignal, \noisySignal | \distributionIndex'}^\seqLength
  \distributionIndexWeight'(\distributionIndex').
\]
Consequently, defining $\universalDist_{\noiselessSignal^\seqLength, \noisySignal^\seqLength}$ as in \eqref{eq:universal-dist} based on $(\jointDist_{\noiselessSignal, \noisySignal | \distributionIndex})_{\distributionIndex \in \distributionIndexDomain}$ is equivalent to using the family $(\jointDist_{\noiselessSignal, \noisySignal | \distributionIndex'})_{\distributionIndex' \in \distributionIndexDomain'}$ instead. Therefore, in the following we assume without loss of generality that \eqref{eq:consistency-stronger} holds.

The following lemma contains the typicality argument used to establish Theorem~\ref{theorem:iid-denoising}.
\begin{lemma}
\label{lemma:kldiv-iid-detail}
If $\cardinality{\noisySignalAlphabet} < \infty$, then for every $\typicalityParameter \in (0,\infty)$, we have
\begin{multline*}
\log\frac{1}{\distributionIndexWeight(\distributionIndex_0)}
-
\kldiv{\jointDist_{\noisySignal}^\seqLength}{\universalDist_{\noisySignal^\seqLength}}
\leq
\frac{1}{\distributionIndexWeight(\distributionIndex_0)}
\left(
  \exp\left(
    -
    2
    \seqLength^{\frac{1}{2}}
  \right)
  +
  \distributionIndexWeight(\complement{\distributionIndexDomain_\seqLength} \setminus \{\distributionIndex_0\})
\right)
\\
+
2
\cardinality{\noisySignalAlphabet}
\left(
  \log\frac{1}{\distributionIndexWeight(\distributionIndex_0)}
  +
  \seqLength
  \log\frac{1}{\jointDistMin{\distributionIndex_0}}
\right)
\exp(-2\seqLength\typicalityParameter^2),
\end{multline*}
where
\begin{equation}
\label{eq:vanishing-parameter-set}
\distributionIndexDomain_\seqLength
:=
\left\{
  \distributionIndex \in \distributionIndexDomain
  :~
  \tvdist{
    \jointDist_{\noisySignal}
    -
    \jointDist_{\noisySignal | \distributionIndex}
  }
  -
  \cardinality{\noisySignalAlphabet}
  \typicalityParameter
  \geq
  \seqLength^{-\frac{1}{4}}
\right\}
\end{equation}
and
\begin{equation}
\label{eq:min-prob}
\jointDistMin{\distributionIndex}
:=
\min_{\substack{\noisySignalValue \in \noisySignalAlphabet \\ \jointDist_{\noisySignal | \distributionIndex}(\noisySignalValue) > 0}} \jointDist_{\noisySignal | \distributionIndex}(\noisySignalValue).
\end{equation}
\end{lemma}
\begin{proof}
We define a typical set
\begin{equation}
\label{eq:typical-set}
\typicalSet
:=
\left\{
  \noisySignalValue^\seqLength \in \noiselessSignalAlphabet^\seqLength
  ~:~
  \forall \noisySignalValue \in \noisySignalAlphabet
  \absoluteValue{
    \frac{1}{\seqLength}
    \countingFunction(\noisySignalValue^\seqLength, \noisySignalValue) - \jointDist_{\noisySignal}(\noisySignalValue)
  }
  <
  \typicalityParameter
\right\}
\end{equation}
and bound
\begin{align}
\log\frac{1}{\distributionIndexWeight(\distributionIndex_0)}
-
\kldiv{\jointDist_{\noisySignal}^\seqLength}{\universalDist_{\noisySignal^\seqLength}}
&=
\Expectation_{\jointDist_{\noisySignal}^\seqLength}\left(
  \log\frac{
             \universalDist_{\noisySignal^\seqLength}(\noisySignal^\seqLength)
           }{
             \distributionIndexWeight(\distributionIndex_0)
             \jointDist_{\noisySignal}^\seqLength(\noisySignal^\seqLength)
           }
\right)
\\
\label{eq:typicality-split}
&=
\Expectation_{\jointDist_{\noisySignal}^\seqLength}\left(
  \indicator{\noisySignal^\seqLength \in \typicalSet}
  \log\frac{
             \universalDist_{\noisySignal^\seqLength}(\noisySignal^\seqLength)
           }{
             \distributionIndexWeight(\distributionIndex_0)
             \jointDist_{\noisySignal}^\seqLength(\noisySignal^\seqLength)
           }
\right)
+
\Expectation_{\jointDist_{\noisySignal}^\seqLength}\left(
  \indicator{\noisySignal^\seqLength \in \complement{\typicalSet}}
  \log\frac{
             \universalDist_{\noisySignal^\seqLength}(\noisySignal^\seqLength)
           }{
             \distributionIndexWeight(\distributionIndex_0)
             \jointDist_{\noisySignal}^\seqLength(\noisySignal^\seqLength)
           }
\right).
\end{align}

We next bound the summands that appear in \eqref{eq:typicality-split} separately. For the second summand in \eqref{eq:typicality-split}, we have
\begin{align}
\nonumber
\Expectation_{\jointDist_{\noisySignal}^\seqLength}\left(
  \indicator{\noisySignal^\seqLength \in \complement{\typicalSet}}
  \log\frac{
             \universalDist_{\noisySignal^\seqLength}(\noisySignal^\seqLength)
           }{
             \distributionIndexWeight(\distributionIndex_0)
             \jointDist_{\noisySignal}^\seqLength(\noisySignal^\seqLength)
           }
\right)
\overset{(a)}&{\leq}
\Expectation_{\jointDist_{\noisySignal}^\seqLength}\left(
  \indicator{\noisySignal^\seqLength \in \complement{\typicalSet}}
  \log\frac{
             1
           }{
             \distributionIndexWeight(\distributionIndex_0)
             \jointDistMin{\distributionIndex_0}^\seqLength
           }
\right)
\\
\nonumber
&=
\jointDist_{\noisySignal}^\seqLength(\complement{\typicalSet})
\log\frac{
            1
          }{
            \distributionIndexWeight(\distributionIndex_0)
            \jointDistMin{\distributionIndex_0}^\seqLength
          }
\\
\label{eq:atypical-term}
\overset{(b)}&{\leq}
2
\cardinality{\noisySignalAlphabet}
\left(
  \log\frac{1}{\distributionIndexWeight(\distributionIndex_0)}
  +
  \seqLength
  \log\frac{1}{\jointDistMin{\distributionIndex_0}}
\right)
\exp(-2\seqLength\typicalityParameter^2)
\end{align}
where in step (a) we have used $\universalDist_{\noisySignal^\seqLength}(\noisySignal^\seqLength) \leq 1$ and \eqref{eq:min-prob}, and in step (b) we have used \cite[Lemma 2.12]{csiszar2011information} along with the refinement pointed out in the unnumbered remark following the lemma.

For the first summand in \eqref{eq:typicality-split}, we have
\begin{align}
\nonumber
\Expectation_{\jointDist_{\noisySignal}^\seqLength}\left(
  \indicator{\noisySignal^\seqLength \in \typicalSet}
  \log\frac{
             \universalDist_{\noisySignal^\seqLength}(\noisySignal^\seqLength)
           }{
             \distributionIndexWeight(\distributionIndex_0)
             \jointDist_{\noisySignal}^\seqLength(\noisySignal^\seqLength)
           }
\right)
\overset{\eqref{eq:universal-dist}}&{=}
\Expectation_{\jointDist_{\noisySignal}^\seqLength}\left(
  \indicator{\noisySignal^\seqLength \in \typicalSet}
  \log\frac{
             \sum_{\distributionIndex \in \distributionIndexDomain}
               \distributionIndexWeight(\distributionIndex)
               \jointDist_{\noisySignal | \distributionIndex}^\seqLength(\noisySignal^\seqLength)
           }{
             \distributionIndexWeight(\distributionIndex_0)
             \jointDist_{\noisySignal}^\seqLength(\noisySignal^\seqLength)
           }
\right)
\\
\nonumber
&=
\Expectation_{\jointDist_{\noisySignal}^\seqLength}\left(
  \indicator{\noisySignal^\seqLength \in \typicalSet}
  \log\left(
    1
    +
    \frac{
           \sum_{\substack{\distributionIndex \in \distributionIndexDomain \\ \distributionIndex \neq \distributionIndex_0}}
             \distributionIndexWeight(\distributionIndex)
             \jointDist_{\noisySignal | \distributionIndex}^\seqLength(\noisySignal^\seqLength)
         }{
           \distributionIndexWeight(\distributionIndex_0)
           \jointDist_{\noisySignal}^\seqLength(\noisySignal^\seqLength)
         }
  \right)
\right)
\\
\nonumber
\overset{(a)}&{\leq}
\Expectation_{\jointDist_{\noisySignal}^\seqLength}\left(
  \indicator{\noisySignal^\seqLength \in \typicalSet}
  \frac{
          \sum_{\substack{\distributionIndex \in \distributionIndexDomain \\ \distributionIndex \neq \distributionIndex_0}}
            \distributionIndexWeight(\distributionIndex)
            \jointDist_{\noisySignal | \distributionIndex}^\seqLength(\noisySignal^\seqLength)
        }{
          \distributionIndexWeight(\distributionIndex_0)
          \jointDist_{\noisySignal}^\seqLength(\noisySignal^\seqLength)
        }
\right)
\\
\nonumber
\overset{(b)}&{=}
\sum_{\substack{\distributionIndex \in \distributionIndexDomain \\ \distributionIndex \neq \distributionIndex_0}}
  \Expectation_{\jointDist_{\noisySignal}^\seqLength}\left(
    \indicator{\noisySignal^\seqLength \in \typicalSet}
    \frac{
              \distributionIndexWeight(\distributionIndex)
              \jointDist_{\noisySignal | \distributionIndex}^\seqLength(\noisySignal^\seqLength)
          }{
            \distributionIndexWeight(\distributionIndex_0)
            \jointDist_{\noisySignal}^\seqLength(\noisySignal^\seqLength)
          }
  \right)
\\
\nonumber
\overset{(c)}&{=}
\sum_{\substack{\distributionIndex \in \distributionIndexDomain \\ \distributionIndex \neq \distributionIndex_0}}
  \Expectation_{\jointDist_{\noisySignal | \distributionIndex}^\seqLength}\left(
    \indicator{\noisySignal^\seqLength \in \typicalSet}
    \frac{
              \distributionIndexWeight(\distributionIndex)
          }{
            \distributionIndexWeight(\distributionIndex_0)
          }
  \right)
\\
\nonumber
&=
\frac{1}{\distributionIndexWeight(\distributionIndex_0)}
\sum_{\substack{\distributionIndex \in \distributionIndexDomain \\ \distributionIndex \neq \distributionIndex_0}}
  \distributionIndexWeight(\distributionIndex)
  \jointDist_{\noisySignal | \distributionIndex}^\seqLength(\typicalSet)
\\
\label{eq:typical-term-1}
&\leq
\frac{1}{\distributionIndexWeight(\distributionIndex_0)}
\left(
  \sum_{\substack{\distributionIndex \in \distributionIndexDomain_\seqLength}}
    \distributionIndexWeight(\distributionIndex)
    \jointDist_{\noisySignal | \distributionIndex}^\seqLength(\typicalSet)
  +
  \distributionIndexWeight\left(\complement{\distributionIndexDomain_\seqLength} \setminus \{\distributionIndex_0\}\right)
\right)
\end{align}
where in step (a) we have used the inequality $\log(1+\generalReal) \leq \generalReal$ for all $\generalReal \in (-1,\infty)$, in (b) we have used linearity of expectation, and in (c) we have used change of measure.

For $\distributionIndex \in \distributionIndexDomain_\seqLength$, we have
\begin{align}
\nonumber
&\hphantom{{}={}}
\jointDist_{\noisySignal | \distributionIndex}^\seqLength(\typicalSet)
\\
\nonumber
\overset{\eqref{eq:typical-set}}&{=}
\jointDist_{\noisySignal | \distributionIndex}^\seqLength\left(
  \left\{
    \noisySignalValue^\seqLength \in \noisySignalAlphabet^\seqLength
    ~:~
    \forall \noisySignalValue \in \noisySignalAlphabet
    \absoluteValue{
      \frac{1}{\seqLength}
      \countingFunction(\noisySignalValue^\seqLength, \noisySignalValue) - \jointDist_{\noisySignal}(\noisySignalValue)
    }
    <
    \typicalityParameter_\seqLength
  \right\}
\right)
\\
\nonumber
\overset{(a)}&{\leq}
\jointDist_{\noisySignal | \distributionIndex}^\seqLength\left(
  \left\{
    \noisySignalValue^\seqLength \in \noisySignalAlphabet^\seqLength
    ~:~
    \forall \noisySignalValue \in \tvDistMaximizer
    \absoluteValue{
      \frac{1}{\seqLength}
      \countingFunction(\noisySignalValue^\seqLength, \noisySignalValue) - \jointDist_{\noisySignal}(\noisySignalValue)
    }
    <
    \typicalityParameter_\seqLength
  \right\}
\right)
\\
\nonumber
&\leq
\jointDist_{\noisySignal | \distributionIndex}^\seqLength\left(
  \left\{
    \noisySignalValue^\seqLength \in \noisySignalAlphabet^\seqLength
    ~:~
    \sum_{\noisySignalValue \in \tvDistMaximizer}
      \absoluteValue{
        \frac{1}{\seqLength}
        \countingFunction(\noisySignalValue^\seqLength, \noisySignalValue) - \jointDist_{\noisySignal}(\noisySignalValue)
      }
    <
    \cardinality{\tvDistMaximizer}
    \typicalityParameter_\seqLength
  \right\}
\right)
\\
\nonumber
&\leq
\jointDist_{\noisySignal | \distributionIndex}^\seqLength\left(
  \left\{
    \noisySignalValue^\seqLength \in \noisySignalAlphabet^\seqLength
    ~:~
    \absoluteValue{
      \sum_{\noisySignalValue \in \tvDistMaximizer}\left(
        \frac{1}{\seqLength}
        \countingFunction(\noisySignalValue^\seqLength, \noisySignalValue) - \jointDist_{\noisySignal}(\noisySignalValue)
      \right)
    }
    <
    \cardinality{\tvDistMaximizer}
    \typicalityParameter_\seqLength
  \right\}
\right)
\\
\nonumber
&=
\jointDist_{\noisySignal | \distributionIndex}^\seqLength\left(
  \left\{
    \noisySignalValue^\seqLength \in \noisySignalAlphabet^\seqLength
    ~:~
    \absoluteValue{
      \sum_{\noisySignalValue \in \tvDistMaximizer}\left(
        \frac{1}{\seqLength}
        \countingFunction(\noisySignalValue^\seqLength, \noisySignalValue) - \jointDist_{\noisySignal | \distributionIndex}(\noisySignalValue)
      \right)
      -
      \sum_{\noisySignalValue \in \tvDistMaximizer}\left(
        \jointDist_{\noisySignal}(\noisySignalValue)
        -
        \jointDist_{\noisySignal | \distributionIndex}(\noisySignalValue)
      \right)
    }
    <
    \cardinality{\tvDistMaximizer}
    \typicalityParameter_\seqLength
  \right\}
\right)
\\
\nonumber
&=
\jointDist_{\noisySignal | \distributionIndex}^\seqLength\left(
  \left\{
    \noisySignalValue^\seqLength \in \noisySignalAlphabet^\seqLength
    ~:~
    \absoluteValue{
      \sum_{\noisySignalValue \in \tvDistMaximizer}\left(
        \frac{1}{\seqLength}
        \countingFunction(\noisySignalValue^\seqLength, \noisySignalValue) - \jointDist_{\noisySignal | \distributionIndex}(\noisySignalValue)
      \right)
      -
      \tvdist{
        \jointDist_{\noisySignal}
        -
        \jointDist_{\noisySignal | \distributionIndex}
      }
    }
    <
    \cardinality{\tvDistMaximizer}
    \typicalityParameter_\seqLength
  \right\}
\right)
\\
\nonumber
&\leq
\jointDist_{\noisySignal | \distributionIndex}^\seqLength\left(
  \left\{
    \noisySignalValue^\seqLength \in \noisySignalAlphabet^\seqLength
    ~:~
    -
    \sum_{\noisySignalValue \in \tvDistMaximizer}\left(
      \frac{1}{\seqLength}
      \countingFunction(\noisySignalValue^\seqLength, \noisySignalValue) - \jointDist_{\noisySignal | \distributionIndex}(\noisySignalValue)
    \right)
    +
    \tvdist{
      \jointDist_{\noisySignal}
      -
      \jointDist_{\noisySignal | \distributionIndex}
    }
    <
    \cardinality{\tvDistMaximizer}
    \typicalityParameter_\seqLength
  \right\}
\right)
\\
\nonumber
&=
\jointDist_{\noisySignal | \distributionIndex}^\seqLength\left(
  \left\{
    \noisySignalValue^\seqLength \in \noisySignalAlphabet^\seqLength
    ~:~
    \frac{1}{\seqLength}
    \sum_{\noisySignalValue \in \tvDistMaximizer}
      \countingFunction(\noisySignalValue^\seqLength, \noisySignalValue)
    -
    \jointDist_{\noisySignal | \distributionIndex}(\tvDistMaximizer)
    >
    \tvdist{
      \jointDist_{\noisySignal}
      -
      \jointDist_{\noisySignal | \distributionIndex}
    }
    -
    \cardinality{\tvDistMaximizer}
    \typicalityParameter_\seqLength
  \right\}
\right)
\\
\nonumber
\overset{(b)}&{\leq}
\exp\left(
  -
  2
  \seqLength
  \left(
    \tvdist{
      \jointDist_{\noisySignal}
      -
      \jointDist_{\noisySignal | \distributionIndex}
    }
    -
    \cardinality{\tvDistMaximizer}
    \typicalityParameter_\seqLength
  \right)^2
\right)
\\
\label{eq:typical-term-2}
&\leq
\exp\left(
  -
  2
  \seqLength
  \left(
    \tvdist{
      \jointDist_{\noisySignal}
      -
      \jointDist_{\noisySignal | \distributionIndex}
    }
    -
    \cardinality{\noisySignalAlphabet}
    \typicalityParameter_\seqLength
  \right)^2
\right),
\end{align}
where in step (a) we introduce $\tvDistMaximizer \subset \noisySignalAlphabet$ with the property $\jointDist_{\noisySignal}(\tvDistMaximizer) - \jointDist_{\noisySignal | \distributionIndex}(\tvDistMaximizer) = \tvdist{\jointDist_{\noisySignal} - \jointDist_{\noisySignal | \distributionIndex}}$ (the existence of such a set is ensured by the definition of variational distance), and step (b) holds due to Hoeffding's inequality (e.g., \cite[Theorem 2.8]{boucheron2013concentration}) because
$
\sum_{\noisySignalValue \in \tvDistMaximizer}
  \countingFunction(\noisySignalValue^\seqLength, \noisySignalValue)
$
is a binomial random variable with $\seqLength$ trials and success probability $\jointDist_{\noisySignal | \distributionIndex}(\tvDistMaximizer)$.

We substitute \eqref{eq:typical-term-2} into \eqref{eq:typical-term-1} and obtain

\begin{align}
\nonumber
&\hphantom{{}={}}
\Expectation_{\jointDist_{\noisySignal}^\seqLength}\left(
  \indicator{\noisySignal^\seqLength \in \typicalSet}
  \log\frac{
             \universalDist_{\noisySignal^\seqLength}(\noisySignal^\seqLength)
           }{
             \distributionIndexWeight(\distributionIndex_0)
             \jointDist_{\noisySignal}^\seqLength(\noisySignal^\seqLength)
           }
\right)
\\
\nonumber
&\leq
\frac{1}{\distributionIndexWeight(\distributionIndex_0)}
\left(
  \sum_{\substack{\distributionIndex \in \distributionIndexDomain_\seqLength}}
    \distributionIndexWeight(\distributionIndex)
    \exp\left(
      -
      2
      \seqLength
      \left(
        \tvdist{
          \jointDist_{\noisySignal}
          -
          \jointDist_{\noisySignal | \distributionIndex}
        }
        -
        \cardinality{\noisySignalAlphabet}
        \typicalityParameter_\seqLength
      \right)^2
    \right)
  +
  \distributionIndexWeight(\complement{\distributionIndexDomain_\seqLength} \setminus \{\distributionIndex_0\})
\right)
\\
\nonumber
&\leq
\frac{1}{\distributionIndexWeight(\distributionIndex_0)}
\left(
  \sum_{\substack{\distributionIndex \in \distributionIndexDomain_\seqLength}}
    \distributionIndexWeight(\distributionIndex)
    \exp\left(
      -
      2
      \seqLength^\frac{1}{2}
    \right)
  +
  \distributionIndexWeight(\complement{\distributionIndexDomain_\seqLength} \setminus \{\distributionIndex_0\})
\right)
\\
\label{eq:typical-term-3}
&\leq
\frac{1}{\distributionIndexWeight(\distributionIndex_0)}
\left(
  \exp\left(
    -
    2
    \seqLength^\frac{1}{2}
  \right)
  +
  \distributionIndexWeight(\complement{\distributionIndexDomain_\seqLength} \setminus \{\distributionIndex_0\})
\right)
\end{align}

The lemma follows by substituting \eqref{eq:atypical-term} and \eqref{eq:typical-term-3} into \eqref{eq:typicality-split}.
\end{proof}

If $\distributionIndexDomain$ and $\noisySignalAlphabet$ are finite, we bound
\begin{align}
\nonumber
\kldiv{\jointDist_{\noiselessSignal, \noisySignal}^\seqLength}{\universalDist_{\noiselessSignal^\seqLength, \noisySignal^\seqLength}}
-
\kldiv{\jointDist_{\noisySignal}^\seqLength}{\universalDist_{\noisySignal^\seqLength}}
\overset{(a)}&{=}
\Expectation_{\jointDist_{\noiselessSignal, \noisySignal}^\seqLength}
  \log\frac{
             \jointDist_{\noiselessSignal, \noisySignal}^\seqLength(\noiselessSignal^\seqLength, \noisySignal^\seqLength)
           }{
             \universalDist_{\noiselessSignal^\seqLength, \noisySignal^\seqLength}(\noiselessSignal^\seqLength, \noisySignal^\seqLength)
           }
-
\kldiv{\jointDist_{\noisySignal}^\seqLength}{\universalDist_{\noisySignal^\seqLength}}
\\
\label{eq:kldiff-bound}
\overset{(b)}&{\leq}
\log\frac{1}{\distributionIndexWeight(\distributionIndex_0)}
-
\kldiv{\jointDist_{\noisySignal}^\seqLength}{\universalDist_{\noisySignal^\seqLength}}
\end{align}
where in step (a) we have substituted the definition of Kullback-Leibler divergence, for (b) we have used that due to \eqref{eq:universal-dist}, it holds that
\[
\universalDist_{\noiselessSignal^\seqLength, \noisySignal^\seqLength}(\noiselessSignal^\seqLength, \noisySignal^\seqLength)
\geq
\distributionIndexWeight(\distributionIndex_0)
\jointDist_{\noiselessSignal, \noisySignal}^\seqLength(\noiselessSignalValue^\seqLength, \noisySignalValue^\seqLength)
\]
pointwise. We bound this term further by invoking Lemma~\ref{lemma:kldiv-iid-detail} with $\typicalityParameter := \seqLength^{-\frac{1}{4}}$. Substituting this choice into the condition in \eqref{eq:vanishing-parameter-set} and rearranging terms, we get
\[
  \seqLength^{\frac{1}{4}}
  \geq
  \frac{1 + \cardinality{\noisySignalAlphabet}}{
    \tvdist{
      \jointDist_{\noisySignal}
      -
      \jointDist_{\noisySignal | \distributionIndex}
    }
  }.
\]
Hence $\distributionIndexDomain_\seqLength = \distributionIndexDomain \setminus \{\distributionIndex_0\}$ for all
\begin{equation}
\label{eq:seqlength-condition}
  \seqLength
  \geq
  \left(
    \frac{1 + \cardinality{\noisySignalAlphabet}}{
      \min_{\distributionIndex \in \distributionIndexDomain \setminus \{\distributionIndex_0\}}
      \tvdist{
        \jointDist_{\noisySignal}
        -
        \jointDist_{\noisySignal | \distributionIndex}
      }
    }
  \right)^4
\end{equation}
So for all $\seqLength$ that satisfy \eqref{eq:seqlength-condition}, we have
\begin{equation*}
\kldiv{\jointDist_{\noiselessSignal, \noisySignal}^\seqLength}{\universalDist_{\noiselessSignal^\seqLength, \noisySignal^\seqLength}}
-
\kldiv{\jointDist_{\noisySignal}^\seqLength}{\universalDist_{\noisySignal^\seqLength}}
\leq
\exp\left(
  -
  2
  \seqLength^{\frac{1}{2}}
\right)
\left(
  \frac{1}{\distributionIndexWeight(\distributionIndex_0)}
  +
  2
  \cardinality{\noisySignalAlphabet}
  \left(
    \log\frac{1}{\distributionIndexWeight(\distributionIndex_0)}
    +
    \seqLength
    \log\frac{1}{\jointDistMin{\distributionIndex_0}}
  \right)
\right).
\end{equation*}
For large enough $\seqLength$, this is upper bounded by $\exp(-\sqrt{T})$, proving the theorem for this case.

Next, we consider the case that $\distributionIndexDomain$ is countably infinite and $(\noisySignalAlphabet, \noisySignalSigmaAlgebra)$ is a general measurable space. We will first construct a sequence of finite alphabets $(\noisySignalAlphabet_\seqLength)_{\seqLength \in \naturals}$ and a corresponding sequence of quantization maps $\quantizer_\seqLength: \noisySignalAlphabet \rightarrow \noisySignalAlphabet_\seqLength$. For every $\distributionIndex \in \distributionIndexDomain \setminus \{\distributionIndex_0\}$, due to \eqref{eq:consistency-stronger}, we can choose some set $\quantizationSetIntermediate_\distributionIndex \in \noiselessSignalSigmaAlgebra$ with $\jointDist_{\noiselessSignal | \distributionIndex}(\quantizationSetIntermediate_\distributionIndex) \neq \jointDist_{\noiselessSignal}(\quantizationSetIntermediate_\distributionIndex)$. Then let $\quantizationSet_\distributionIndex := \quantizationSetIntermediate_\distributionIndex$ if $\jointDist_{\noiselessSignal}(\quantizationSetIntermediate_\distributionIndex) > 0$ and let $\quantizationSet_\distributionIndex := \noisySignalAlphabet \setminus \quantizationSetIntermediate_\distributionIndex$ otherwise. We have thus constructed a family $(\quantizationSet_\distributionIndex)_{\distributionIndex \in \distributionIndexDomain \setminus \{\distributionIndex_0\}}$ of measurable subsets of $\noisySignalAlphabet$ with the property
\begin{equation}
\label{eq:quantizer-property}
\forall \distributionIndex \in \distributionIndexDomain \setminus \{\distributionIndex_0\}:~
\jointDist_{\noisySignal}(\quantizationSet_\distributionIndex) > 0,~
\jointDist_{\noisySignal}(\quantizationSet_\distributionIndex) \neq \jointDist_{\noisySignal | \distributionIndex}(\quantizationSet_\distributionIndex).
\end{equation}
We next define a sequence $(\quantizerLength_\seqLength)_{\seqLength \in \naturals}$ via
\begin{equation}
\label{eq:quantizer-length}
\quantizerLength_\seqLength
:=
\max\left\{
  \quantizerLength
  \in
  \left\{
    1,
    \dots,
    \floor{\frac{1}{8} \log_2 \seqLength}
  \right\}
  :~
  \min_{(\quantizerBit_1, \dots, \quantizerBit_\quantizerLength) \in \minset_\quantizerLength}
    \jointDist_{\noisySignal}\left(
      \quantizationSet_{\distributionIndex_1}[\quantizerBit_1]
      \cap
      \cdots
      \cap
      \quantizationSet_{\distributionIndex_\quantizerLength}[\quantizerBit_\quantizerLength]
    \right)
  \geq
  \exp(-\seqLength)
\right\},
\end{equation}
where $(\distributionIndex_\distributionIndexIndex)_{\distributionIndexIndex \in \naturals}$ is an enumeration of $\distributionIndexDomain \setminus \{\distributionIndex_0\}$,
\[
\quantizationSet_{\distributionIndex}[\quantizerBit]
:=
\begin{cases}
  \quantizationSet_{\distributionIndex}, &\quantizerBit=1 \\
  \complement{\quantizationSet_{\distributionIndex}}, &\text{otherwise,}
\end{cases}
\]
and
\[
\minset_\quantizerLength
:=
\left\{
  (\quantizerBit_1, \dots, \quantizerBit_\quantizerLength)
  \in
  \{0,1\}^\quantizerLength
  :~
  \jointDist_{\noisySignal}\left(
    \quantizationSet_{\distributionIndex_1}[\quantizerBit_1]
    \cap
    \cdots
    \cap
    \quantizationSet_{\distributionIndex_\quantizerLength}[\quantizerBit_\quantizerLength]
  \right)
  >
  0
\right\}.
\]
\begin{lemma}
\label{lemma:minset-nonempty}
For every $\quantizerLength \in \naturals$, we have $\minset_\quantizerLength \neq \emptyset$.
\end{lemma}
\begin{proof}
We prove the lemma by induction. If $\quantizerLength = 1$, then $1 \in \minset_\quantizerLength$ by property \eqref{eq:quantizer-property}. If $\quantizerLength > 1$, we invoke the induction hypothesis to obtain $(\quantizerBit_1, \dots, \quantizerBit_{\quantizerLength-1}) \in \minset_{\quantizerLength-1}$, i.e.,
$
\jointDist_{\noisySignal}\left(
  \quantizationSet_{\distributionIndex_1}[\quantizerBit_1]
  \cap
  \cdots
  \cap
  \quantizationSet_{\distributionIndex_{\quantizerLength-1}}[\quantizerBit_{\quantizerLength-1}]
\right)
>
0
$.
If
$
\jointDist_{\noisySignal}\left(
  \quantizationSet_{\distributionIndex_1}[\quantizerBit_1]
  \cap
  \cdots
  \cap
  \quantizationSet_{\distributionIndex_{\quantizerLength-1}}[\quantizerBit_{\quantizerLength-1}]
  \cap
  \quantizationSet_{\distributionIndex_{\quantizerLength}}
\right)
>
0
$,
then we have $(\quantizerBit_1, \dots, \quantizerBit_{\quantizerLength-1}, 1) \in \minset_\quantizerLength$. If, on the other hand,
$
\jointDist_{\noisySignal}\left(
  \quantizationSet_{\distributionIndex_1}[\quantizerBit_1]
  \cap
  \cdots
  \cap
  \quantizationSet_{\distributionIndex_{\quantizerLength-1}}[\quantizerBit_{\quantizerLength-1}]
  \cap
  \quantizationSet_{\distributionIndex_{\quantizerLength}}
\right)
=
0
$,
we have
\begin{align*}
&\hphantom{{}={}}
\jointDist_{\noisySignal}\left(
  \quantizationSet_{\distributionIndex_1}[\quantizerBit_1]
  \cap
  \cdots
  \cap
  \quantizationSet_{\distributionIndex_{\quantizerLength-1}}[\quantizerBit_{\quantizerLength-1}]
  \cap
  \complement{\quantizationSet_{\distributionIndex_{\quantizerLength}}}
\right)
\\
&=
\jointDist_{\noisySignal}\left(
  \left(
    \quantizationSet_{\distributionIndex_1}[\quantizerBit_1]
    \cap
    \cdots
    \cap
    \quantizationSet_{\distributionIndex_{\quantizerLength-1}}[\quantizerBit_{\quantizerLength-1}]
  \right)
  \setminus
  \left(
    \quantizationSet_{\distributionIndex_1}[\quantizerBit_1]
    \cap
    \cdots
    \cap
    \quantizationSet_{\distributionIndex_{\quantizerLength-1}}[\quantizerBit_{\quantizerLength-1}]
    \cap
    \quantizationSet_{\distributionIndex_{\quantizerLength}}
  \right)
\right)
\\
&=
\jointDist_{\noisySignal}\left(
  \quantizationSet_{\distributionIndex_1}[\quantizerBit_1]
  \cap
  \cdots
  \cap
  \quantizationSet_{\distributionIndex_{\quantizerLength-1}}[\quantizerBit_{\quantizerLength-1}]
\right)
-
\jointDist_{\noisySignal}\left(
  \quantizationSet_{\distributionIndex_1}[\quantizerBit_1]
  \cap
  \cdots
  \cap
  \quantizationSet_{\distributionIndex_{\quantizerLength-1}}[\quantizerBit_{\quantizerLength-1}]
  \cap
  \quantizationSet_{\distributionIndex_{\quantizerLength}}
\right)
\\
&=
\jointDist_{\noisySignal}\left(
  \quantizationSet_{\distributionIndex_1}[\quantizerBit_1]
  \cap
  \cdots
  \cap
  \quantizationSet_{\distributionIndex_{\quantizerLength-1}}[\quantizerBit_{\quantizerLength-1}]
\right)
\\
&>
0
\end{align*}
and hence $(\quantizerBit_1, \dots, \quantizerBit_{\quantizerLength-1}, 0) \in \minset_\quantizerLength$.
\end{proof}

\begin{lemma}
\label{lemma:seqlength}
$(\quantizerLength_\seqLength)_{\seqLength \in \naturals}$ is a nondecreasing sequence with $\lim_{\seqLength\rightarrow\infty} \quantizerLength_\seqLength = \infty$.
\end{lemma}
\begin{proof}
It is clear from definition \eqref{eq:quantizer-length} that the set of which $\quantizerLength_{\seqLength+1}$ is the maximum is a superset of the set of which $\quantizerLength_\seqLength$ is the maximum and therefore $\quantizerLength_{\seqLength+1} \geq \quantizerLength_\seqLength$. Assume, towards a contradiction, that $\lim_{\seqLength\rightarrow\infty} \quantizerLength_\seqLength \neq \infty$. Since we have already shown that the sequence is nondecreasing, this implies that there is some $\quantizerLength_{\max} \in \naturals$ such that $\quantizerLength_\seqLength = \quantizerLength_{\max}$ for all sufficiently large $\seqLength \in \naturals$. Examining definition \eqref{eq:quantizer-length} once more, we can see that it implies
\begin{equation}
\label{eq:quantizer-length-implication}
  \min_{(\quantizerBit_1, \dots, \quantizerBit_{\quantizerLength_{\max}+1}) \in \minset_{\quantizerLength_{\max}+1}}
    \jointDist_{\noisySignal}\left(
      \quantizationSet_{\distributionIndex_1}[\quantizerBit_1])
      \cap
      \cdots
      \cap
      \quantizationSet_{\distributionIndex_\quantizerLength}[\quantizerBit_{\quantizerLength_{\max}+1}]
    \right)
  <
  \exp(-\seqLength)
\end{equation}
for all sufficiently large $\seqLength$. Since the left hand side is independent of $\seqLength$ and due to Lemma~\ref{lemma:minset-nonempty} $\minset_{\quantizerLength_{\max}+1} \neq \emptyset$, this means that there is $(\quantizerBit_1, \dots, \quantizerBit_{\quantizerLength_{\max}+1}) \in \{0,1\}^{\quantizerLength_{\max}+1}$ with
$
\jointDist_{\noisySignal}(
  \quantizationSet_{\distributionIndex_1}[\quantizerBit_1]
  \cap
  \cdots
  \cap
  \quantizationSet_{\distributionIndex_\quantizerLength}[\quantizerBit_{\quantizerLength_{\max}+1}]
)
>
0
$
but
$
\jointDist_{\noisySignal}(
  \quantizationSet_{\distributionIndex_1}[\quantizerBit_1]
  \cap
  \cdots
  \cap
  \jointDist_{\noisySignal}(\quantizationSet_{\distributionIndex_\quantizerLength}[\quantizerBit_{\quantizerLength_{\max}+1}]
)
<
\exp(-\seqLength)
$
for all sufficiently large $\seqLength \in \naturals$. This is the contradiction that was needed to prove the lemma.
\end{proof}
For every $\seqLength \in \naturals$, define $\noisySignalAlphabet_\seqLength := \{0,1\}^{\quantizerLength_\seqLength}$ and
\begin{equation}
\label{eq:quantizer}
\quantizer_\seqLength:~
\noisySignalAlphabet \rightarrow \noisySignalAlphabet_\seqLength,~
\noisySignalValue
\mapsto
\left(
  \indicator{\noisySignalValue \in \quantizationSet_{\distributionIndex_1}}, \dots, \indicator{\noisySignalValue \in \quantizationSet_{\distributionIndex_{\quantizerLength_\seqLength}}}
\right).
\end{equation}
We also define
\[
\jointDistMinLen{\seqLength}
:=
\min_{\substack{
  \noisySignalValue_\mathrm{quant} \in \noisySignalAlphabet_\seqLength\\
  \jointDist_{\noisySignal}(\quantizer_\seqLength(\noisySignal) = \noisySignalValue_\mathrm{quant}) > 0
}}
  \jointDist_{\noisySignal}(\quantizer_\seqLength(\noisySignal) = \noisySignalValue_\mathrm{quant})
\]
and note that \eqref{eq:quantizer-length} ensures
\begin{equation}
\label{eq:alphabet-bounds}
\jointDistMinLen{\seqLength} \geq \exp(-\seqLength) \text{ and } \cardinality{\noisySignalAlphabet_\seqLength} \leq 2^{\frac{1}{8} \log_2 \seqLength} = \seqLength^\frac{1}{8}.
\end{equation}
\begin{lemma}
\label{lemma:indexset-quantization}
Defining
\begin{align}
\label{eq:vanishing-parameter-set-infinite}
\distributionIndexDomain_\seqLength
&:=
\left\{
  \distributionIndex \in \distributionIndexDomain
  :~
  \tvdist{
    \jointDist_{\quantizer_\seqLength(\noisySignal)}
    -
    \jointDist_{\quantizer_\seqLength(\noisySignal) | \distributionIndex}
  }
  -
  \cardinality{\noisySignalAlphabet_\seqLength}
  \seqLength^{-\frac{1}{4}}
  \geq
  \seqLength^{-\frac{1}{4}}
\right\}
\\
\nonumber
\distributionIndexDomain_\seqLength'
&:=
\left\{
  \distributionIndex \in \distributionIndexDomain
  :~
  \tvdist{
    \jointDist_{\quantizer_\seqLength(\noisySignal)}
    -
    \jointDist_{\quantizer_\seqLength(\noisySignal) | \distributionIndex}
  }
  -
  \seqLength^{-\frac{1}{8}}
  \geq
  \seqLength^{-\frac{1}{4}}
\right\},
\end{align}
we have $\distributionIndexDomain_\seqLength' \subseteq \distributionIndexDomain_\seqLength$ and
\begin{equation}
\label{eq:indexset-limit}
\bigcup_{\seqLength \in \naturals} \distributionIndexDomain_\seqLength' = \distributionIndexDomain \setminus \{\distributionIndex_0\}.
\end{equation}
\end{lemma}
\begin{proof}
$\distributionIndexDomain_\seqLength' \subseteq \distributionIndexDomain_\seqLength$ follows directly from $\cardinality{\noisySignalAlphabet_\seqLength} \leq \seqLength^\frac{1}{8}$. For \eqref{eq:indexset-limit}, we assume, towards a contradiction, that there is $\quantizerLength \in \naturals$ such that $\distributionIndex_\quantizerLength \notin \distributionIndexDomain_\seqLength'$ for all $\seqLength \in \naturals$. Defining the auxiliary sequence $(\auxSequence_\seqLength)_{\seqLength \in \naturals}$ via
\[
  \auxSequence_\seqLength
  :=
  \tvdist{
    \jointDist_{\quantizer_\seqLength(\noisySignal)}
    -
    \jointDist_{\quantizer_\seqLength(\noisySignal) | \distributionIndex_\quantizerLength}
  },
\]
this means that $\auxSequence_\seqLength < \seqLength^{-\frac{1}{8}} + \seqLength^{-\frac{1}{4}}$ for all $\seqLength \in \naturals$ and in particular implies
$
\lim_{\seqLength \rightarrow \infty}
\auxSequence_\seqLength
=
0
$.
It can be seen in \eqref{eq:quantizer} that $\quantizer_\seqLength(\noisySignal)$ is a function of $\quantizer_{\seqLength+1}(\noisySignal)$ (simply discard the last coordinate), so by the data processing principle, $(\auxSequence_\seqLength)_{\seqLength \in \naturals}$ is a nondecreasing sequence. However, it clearly also is a nonnegative sequence, so the limit statement actually implies $\auxSequence_\seqLength = 0$ for all $\seqLength \in \naturals$. Lemma~\ref{lemma:seqlength} allows us to fix $\seqLength$ with $\quantizerLength_\seqLength \geq \quantizerLength$. $\auxSequence_\seqLength = 0$ and \eqref{eq:quantizer} then imply
\[
\jointDist_{\noisySignal}(\quantizationSet_{\distributionIndex_\quantizerLength})
=
\jointDist_{\quantizer_\seqLength(\noisySignal)}\left(\{0,1\}^{\quantizerLength-1} \times \{1\} \times \{0,1\}^{\quantizerLength_\seqLength-\quantizerLength}\right)
=
\jointDist_{\quantizer_\seqLength(\noisySignal) | \distributionIndex_\quantizerLength}\left(\{0,1\}^{\quantizerLength-1} \times \{1\} \times \{0,1\}^{\quantizerLength_\seqLength-\quantizerLength}\right)
=
\jointDist_{\noisySignal | \distributionIndex_\quantizerLength}(\quantizationSet_{\distributionIndex_\quantizerLength}),
\]
contradicting \eqref{eq:quantizer-property}.
\end{proof}

We now have all the technical ingredient necessary to finish the proof of Lemma~\ref{lemma:kldiv-iid}. Specifically, we can bound
\begin{align}
\nonumber
&\hphantom{{}={}}
\kldiv{\jointDist_{\noiselessSignal, \noisySignal}^\seqLength}{\universalDist_{\noiselessSignal^\seqLength, \noisySignal^\seqLength}}
-
\kldiv{\jointDist_{\noisySignal}^\seqLength}{\universalDist_{\noisySignal^\seqLength}}
\\
\nonumber
\overset{(a)}&{=}
\Expectation_{\jointDist_{\noiselessSignal,\noisySignal}^\seqLength}
  \log\rnderiv{\jointDist_{\noiselessSignal, \noisySignal}^\seqLength}{\universalDist_{\noiselessSignal^\seqLength, \noisySignal^\seqLength}}
    (\noiselessSignal^\seqLength,\noisySignal^\seqLength)
-
\kldiv{\jointDist_{\noisySignal}^\seqLength}{\universalDist_{\noisySignal^\seqLength}}
\\
\nonumber
\overset{\eqref{eq:universal-dist}}&{=}
-
\Expectation_{\jointDist_{\noiselessSignal,\noisySignal}^\seqLength}
  \log\left(
    \sum_{\distributionIndex\in\distributionIndexDomain}\distributionIndexWeight(\distributionIndex)\rnderiv{\jointDist_{\noiselessSignal, \noisySignal | \distributionIndex}^\seqLength}{\jointDist_{\noiselessSignal, \noisySignal}^\seqLength}
      (\noiselessSignal^\seqLength,\noisySignal^\seqLength)
  \right)
-
\kldiv{\jointDist_{\noisySignal}^\seqLength}{\universalDist_{\noisySignal^\seqLength}}
\\
\nonumber
&\leq
-
\Expectation_{\jointDist_{\noiselessSignal,\noisySignal}^\seqLength}
  \log\left(
    \distributionIndexWeight(\distributionIndex_0)\rnderiv{\jointDist_{\noiselessSignal, \noisySignal | \distributionIndex_0}^\seqLength}{\jointDist_{\noiselessSignal, \noisySignal}^\seqLength}
      (\noiselessSignal^\seqLength,\noisySignal^\seqLength)
  \right)
-
\kldiv{\jointDist_{\noisySignal}^\seqLength}{\universalDist_{\noisySignal^\seqLength}}
\\
\nonumber
\overset{(b)}&{\leq}
\log\frac{1}{\distributionIndexWeight(\distributionIndex_0)}
-
\kldiv{\jointDist_{\noisySignal}^\seqLength}{\universalDist_{\noisySignal^\seqLength}}
\\
\nonumber
\overset{(c)}&{\leq}
\log\frac{1}{\distributionIndexWeight(\distributionIndex_0)}
-
\kldiv{\jointDist_{\quantizer_\seqLength(\noisySignal)}^\seqLength}{\universalDist_{\quantizer_\seqLength(\noisySignal^\seqLength)}}
\\
\nonumber
\overset{(d)}&{\leq}
\frac{1}{\distributionIndexWeight(\distributionIndex_0)}
\left(
  \exp\left(
    -
    2
    \seqLength^{\frac{1}{2}}
  \right)
  +
  \distributionIndexWeight(\complement{\distributionIndexDomain_\seqLength} \setminus \{\distributionIndex_0\})
\right)
+
2
\cardinality{\noisySignalAlphabet_\seqLength}
\left(
  \log\frac{1}{\distributionIndexWeight(\distributionIndex_0)}
  +
  \seqLength
  \log\frac{1}{\jointDistMinLen{\seqLength}}
\right)
\exp(-2\seqLength^{\frac{1}{2}})
\\
\label{eq:kldiff-bound-infinite}
\overset{(e)}&{\leq}
\frac{1}{\distributionIndexWeight(\distributionIndex_0)}
\left(
  \exp\left(
    -
    2
    \seqLength^{\frac{1}{2}}
  \right)
  +
  \distributionIndexWeight(\complement{\distributionIndexDomain_\seqLength'} \setminus \{\distributionIndex_0\})
\right)
+
2
\seqLength^{\frac{1}{8}}
\left(
  \log\frac{1}{\distributionIndexWeight(\distributionIndex_0)}
  +
  \seqLength^2
\right)
\exp(-2\seqLength^{\frac{1}{2}})
\end{align}
where step (a) uses the definition of Kullabck-Leibler divergence, step (b) uses $\jointDist_{\noiselessSignal, \noisySignal | \distributionIndex_0}=\jointDist_{\noiselessSignal, \noisySignal}$, (c) follows by the data processing principle (the notation $\quantizer_\seqLength(\noisySignal^\seqLength)$ is to be understood as a component-wise application of the map $\quantizer_\seqLength$), in (d) we have invoked Lemma~\ref{lemma:kldiv-iid-detail} with $\typicalityParameter := \seqLength^{-\frac{1}{4}}$ (note that the definitions \eqref{eq:vanishing-parameter-set} and \eqref{eq:vanishing-parameter-set-infinite} match in this case), and in (e) we have used Lemma~\ref{lemma:indexset-quantization} and \eqref{eq:alphabet-bounds}. It follows from Lemma~\ref{lemma:indexset-quantization} that
\[
\lim_{\seqLength \rightarrow \infty} \distributionIndexWeight(\complement{\distributionIndexDomain_\seqLength'} \setminus \{\distributionIndex_0\})
=
0,
\]
so \eqref{eq:kldiff-bound-infinite} vanishes as $\seqLength \rightarrow \infty$, concluding the proof of Lemma~\ref{lemma:kldiv-iid}.

\subsection{Proof of Theorem~\ref{theorem:iid-denoising-continuous}}
\label{appendix:iid-proof-continuous}

We use the main result of~\cite{clarke2002information} to bound the Kullback-Leibler divergences $\kldiv{\jointDist_{\noiselessSignal, \noisySignal}^\seqLength}{\universalDist_{\noiselessSignal^\seqLength, \noisySignal^\seqLength}}$ and $\kldiv{\jointDist_{\noisySignal}^\seqLength}{\universalDist_{\noisySignal^\seqLength}}$ so we can substitute them into Theorem~\ref{theorem:oneshot}. First, we note that condition \ref{item:cb-smoothness-density} in Theorem~\ref{theorem:iid-denoising-continuous} ensures that~\cite[Condition 1]{clarke2002information} is satisfied for both families $(\jointDist_{\noiselessSignal, \noisySignal | \distributionIndex})_{\distributionIndex \in \distributionIndexDomain}$ and $(\jointDist_{\noisySignal | \distributionIndex})_{\distributionIndex \in \distributionIndexDomain}$ while conditions \ref{item:cb-smoothness-kl} and \ref{item:cb-positive-density} ensure that~\cite[Condition 2]{clarke2002information} is also satisfied. Condition~\ref{item:cb-soundness} corresponds to the soundness condition in~\cite{clarke2002information} which via \cite[Theorem 2.2]{clarke2002information} ensures that also \cite[Condition 3]{clarke2002information} is satisfied. This allows us to invoke~\cite[Theorem 2.1]{clarke2002information} and via the remark after the theorem and condition~\ref{item:cb-fisher-regularity} of Theorem~\ref{theorem:iid-denoising-continuous}, we get the expression of \cite[eq. (1.4)]{clarke2002information}
\begin{align*}
\kldiv{\jointDist_{\noiselessSignal, \noisySignal}^\seqLength}{\universalDist_{\noiselessSignal^\seqLength, \noisySignal^\seqLength}}
&=
\frac{\cbdimension}{2}
\log\frac{\seqLength}{2\pi \eulersNumber}
+
\frac{1}{2} \log \det \fisherInformation_{\noiselessSignal, \noisySignal}(\distributionIndex_0)
+
\landauo(1) \\
\kldiv{\jointDist_{\noisySignal}^\seqLength}{\universalDist_{\noisySignal^\seqLength}}
&=
\frac{\cbdimension}{2}
\log\frac{\seqLength}{2\pi \eulersNumber}
+
\frac{1}{2} \log \det \fisherInformation_{\noisySignal}(\distributionIndex_0)
+
\landauo(1).
\end{align*}
Putting this together and substituting into Theorem~\ref{theorem:oneshot}, we obtain
\[
\regret
=
\kldiv{\jointDist_{\noiselessSignal, \noisySignal}^\seqLength}{\universalDist_{\noiselessSignal^\seqLength, \noisySignal^\seqLength}}
-
\kldiv{\jointDist_{\noisySignal}^\seqLength}{\universalDist_{\noisySignal^\seqLength}}
=
\frac{1}{2} \log \frac{\det \fisherInformation_{\noiselessSignal, \noisySignal}(\distributionIndex_0)}
                      {\det \fisherInformation_{\noisySignal}(\distributionIndex_0)}
+
\landauo(1),
\]
and the theorem follows by dividing both sides by $\seqLength$.

\bibliographystyle{plain}
\bibliography{denoising_references}

\end{document}